\documentclass[10pt,letterpaper,twocolumn]{article}
\usepackage[T1]{fontenc}
\usepackage[utf8]{inputenc}
\usepackage{lmodern}
\usepackage[margin=0.65in,top=0.55in,bottom=0.62in,columnsep=0.24in]{geometry}
\usepackage{amsmath,amssymb,amsthm,mathtools}
\usepackage{graphicx,booktabs,microtype,cite}
\usepackage{titlesec}
\usepackage{tikz}
\usetikzlibrary{arrows.meta,positioning,calc}
\usepackage[hidelinks]{hyperref}
\hypersetup{pdftitle={PPT-assisted capacities of bosonic Gaussian channels},pdfauthor={Matthew Barber and Stefano Pirandola},pdfsubject={Unconstrained and energy-constrained PPT-assisted bosonic quantum capacities}}
\renewcommand{\thesection}{\Roman{section}}
\titleformat{\section}{\normalfont\small\bfseries\centering}{\thesection.}{0.7em}{}
\titlespacing*{\section}{0pt}{14pt plus 3pt minus 2pt}{9pt}
\titleformat{\subsection}{\normalfont\small\bfseries\centering}{\arabic{subsection}.}{0.6em}{}
\titlespacing*{\subsection}{0pt}{12pt}{7pt}
\titleformat{\paragraph}[runin]{\normalfont\itshape}{}{0pt}{}
\titlespacing*{\paragraph}{1em}{5pt}{0.5em}
\newtheorem{theorem}{Theorem}
\newtheorem{appthm}{Theorem}[section]
\newtheorem{applemma}[appthm]{Lemma}
\newtheorem{appprop}[appthm]{Proposition}
\newtheorem{appcor}[appthm]{Corollary}

\newcommand{\id}{\mathrm{id}}

\newcommand{\ket}[1]{\lvert #1\rangle}
\newcommand{\bra}[1]{\langle #1\rvert}
\newcommand{\proj}[1]{\ket{#1}\!\bra{#1}}
\newcommand{\QP}{Q_{\mathrm{PPT}}}
\allowdisplaybreaks[2]
\makeatletter
\renewcommand{\@maketitle}{\newpage\null\vspace{-16pt}\begin{center}{\large\bfseries\@title\par}\vskip10pt{\normalsize\@author\par}\end{center}\vskip8pt}
\makeatother

\title{PPT-assisted capacities of bosonic Gaussian channels}
\author{Matthew Barber and Stefano Pirandola\\[3pt]
\small Department of Computer Science, University of York, York YO10 5GH, United Kingdom}
\date{}
\begin{document}
\twocolumn[\maketitle
\begin{center}\begin{minipage}{0.88\textwidth}\small
We determine the exact quantum communication and entanglement-distribution capacities of
three fundamental bosonic channels, namely thermal loss, thermal amplification and additive
Gaussian noise, assisted by operations that completely preserve positive partial transpose (PPT).
These operations are more powerful than local operations and classical communication (LOCC),
but cannot generate distillable entanglement on their own. For thermal loss, the unconstrained
capacity coincides with the Pirandola--Laurenza--Ottaviani--Banchi (PLOB) bound throughout the
non-entanglement-breaking regime. The corresponding upper bounds are also attained for thermal
amplification and additive noise. We extend these results to energy-constrained communication, establishing exact capacities for all three channels at every prescribed mean input energy. Our results
solve a class of bosonic capacity problems under PPT assistance and establish exact benchmarks
for determining whether the same communication rates can be achieved via LOCC alone.
\end{minipage}\end{center}\vspace{16pt}]

\section{Introduction}\label{sec:intro}
We establish exact PPT-assisted capacities for noisy bosonic Gaussian channels. These channels model optical and microwave attenuation, amplification and random displacements~\cite{r1,r2,r3,r4,r5}. Their capacity problem combines infinite-dimensional inputs with arbitrary adaptive protocols. With local operations and classical communication (LOCC), entanglement distillation and teleportation relate shared entanglement and quantum transmission~\cite{r6,r7,r8,r9,r10}. The Pirandola--Laurenza--Ottaviani--Banchi (PLOB) bound~\cite{r11} determines the unrestricted two-way capacities of pure loss and quantum-limited amplification, but generally leaves a gap for channels with excess noise. Lower bounds have been developed for two-way quantum communication~\cite{PGPBL,r14,r13} and secret-key generation~\cite{PGPBL,OLCSPB,OPB,r14}. Recent photon-sector methods give exact energy-constrained quantum-limited capacities and converse bounds for noisy families~\cite{r15,r16}.

In this paper, we allow completely positive-partial-transpose-preserving operations between channel uses. A state has positive partial transpose (PPT) when transposing one party's indices preserves positivity~\cite{r17,r18}. The allowed operations preserve this property with arbitrary local ancillary systems. They form a strictly larger class than LOCC: they can prepare PPT entangled states, which LOCC cannot prepare from a separable state~\cite{r19,r20,r21,r22}. We prove that they attain the PLOB bound for thermal loss, thermal amplification and additive noise. The same rates govern PPT-assisted entanglement distribution. We then attain the corresponding energy-constrained bounds at every prescribed mean input energy. This closes the coding--converse gap for a nontrivial adaptive bosonic capacity problem and supplies exact benchmarks for locally implementable protocols.

Our central step is an explicit collective PPT distillation protocol. Such a construction is essential because PPT-distillable entanglement can be strictly smaller than its regularized relative-entropy upper bound~\cite{r23}. An exact partial-transpose identity controls the unconstrained construction, while a polynomial norm bound on photon sectors treats fixed energy. Note that the nonseparability of our threshold maps does not exclude other local protocols.

Section~\ref{sec:unconstrained} gives the unconstrained capacities and discusses entanglement distribution and secret key. Section~\ref{sec:energy} establishes the energy-constrained extension, and Sec.~\ref{sec:conclusions} discusses its implications. The Supplemental Material contains the proofs in Appendices~\ref{app:A}--\ref{app:O}, preceded by diagrams of their structure. Appendix~\ref{app:A} fixes the operational definitions; Appendices~\ref{app:B}--\ref{app:H} establish the unconstrained results; Appendices~\ref{app:I} and~\ref{app:J} analyze one-photon resources and the nonseparability of the maps; Appendices~\ref{app:K}--\ref{app:N} establish the energy-constrained results; and Appendix~\ref{app:O} collects benchmarks and a sharper finite-block conjecture.

\section{Unconstrained capacities}\label{sec:unconstrained}
We first state the exact capacities without an input-energy constraint, starting with thermal loss.

\paragraph{Assistance, rates and notation.}
Alice and Bob hold the two sides of a bipartite state. The identity channel is denoted by $\mathrm{id}$. We write $\Gamma_B$ for transposition of Bob's matrix indices in the number basis and use $\Gamma$ when the system is clear. For input systems $A,B$, output systems $A',B'$ and a completely positive trace-preserving map $\Lambda$, define its partial-transpose conjugate by
\begin{equation}\Lambda^\Gamma=\Gamma_{B'}\circ\Lambda\circ\Gamma_B.\label{eq:pptmap}\end{equation}
We call $\Lambda$ completely PPT preserving when $\Lambda^\Gamma$ is also completely positive and trace preserving. Complete positivity includes arbitrary untouched ancillary systems. On bosonic spaces both maps must be normal: their dual maps preserve increasing bounded limits of positive observables. Appendix~\ref{app:A} specifies the bounded or trace-class extensions needed for partial transpose in infinite dimensions.




For a channel $\mathcal N$ from Alice to Bob, $\QP(\mathcal N)$ denotes the quantum capacity with this assistance. A code of length $m$ and dimension $K_{m}$ acts on a $K_{m}$-dimensional system $A_{0}$ held by Alice. Alice and Bob share no entanglement beforehand: any ancillary states start in a separable state that is independent of the input. The protocol interleaves $m$ uses of $\mathcal N$ with completely PPT-preserving operations. It therefore defines a channel $\Lambda_{m}$ from $A_{0}$ to a $K_{m}$-dimensional system $B_{0}$ held by Bob, identified with $A_{0}$ through the number basis. Let $\tilde{A}$ be a reference system of the same dimension, which no operation acts on. The code's worst-case fidelity is
\begin{equation}
    f_{m} = \min_{\Psi}{\bra{\Psi}\left(\id_{\tilde{A}} \otimes \Lambda_{m}\right)\left(\ket{\Psi}\bra{\Psi}\right)\ket{\Psi}}\text{,}
\end{equation}
where the minimum is over all unit vectors $\Psi$ in $\tilde{A} \otimes A_{0}$. A rate $R$ is achievable if, for each $m$, there is a code of length $m$ and dimension $K_{m}$, with worst-case fidelity $f_{m}$, such that $f_m\to1$ and $\liminf_{m \to \infty}m^{-1}\log_2K_m \geq R$. The supremum of these rates is $\QP\left(\mathcal{N}\right)$, in qubits per channel use. 
The capacity $Q_{2}$ uses LOCC assistance instead.


For a mode with nonnegative mean thermal occupation $x$, denote its thermal state by $\tau_x$ and its entropy by $g(x)$. With $j$ the photon number,
\begin{align}
 \tau_x&=\frac{1}{1+x}\sum_{j=0}^{\infty}\left(\frac{x}{1+x}\right)^j|j\rangle\langle j|,\label{eq:thermalstate}\\
 g(x)&=(x+1)\log_2(x+1)-x\log_2x,\quad g(0)=0.\label{eq:entropy}
\end{align}
The zero term is understood by continuity. A channel is entanglement breaking if applying it to one part of any state gives a separable state. Here separable states are trace-norm limits of mixtures of product states.

\paragraph{Thermal loss.}
Let $\eta\in(0,1)$ be the transmissivity and $n\geq0$ the environment's mean photon number. The thermal-loss channel $\mathcal L_{\eta,n}$ mixes an input annihilation operator $a_{\mathrm{in}}$ with an independent environment mode $e$ in $\tau_n$. The output annihilation operator is $b$ and the beam-splitter relation is
\begin{equation}b=\sqrt\eta\,a_{\mathrm{in}}+\sqrt{1-\eta}\,e.\label{eq:losschannel}\end{equation}
The channel is entanglement breaking at and above $n=\eta/(1-\eta)$~\cite{r32}. Let $B_{\mathrm{loss}}(\eta,n)$ denote the PLOB upper bound~\cite{r11}. As shown in Appendices~\ref{app:B}--\ref{app:F}, we have the following exact capacity.
\begin{theorem}[Thermal-loss capacity]\label{thm:loss}
The PPT-assisted quantum capacity is
\begin{equation}
\begin{split}
 &Q_{\mathrm{PPT}}(\mathcal L_{\eta,n})=B_{\mathrm{loss}}(\eta,n)\\
 &=\begin{cases}
 -\log_2[(1-\eta)\eta^n]-g(n),&n<\dfrac{\eta}{1-\eta},\\[3pt]
 0,&n\geq\dfrac{\eta}{1-\eta}.
 \end{cases}
\end{split}\label{eq:loss-capacity}
\end{equation}
Every strictly smaller rate is achievable with fixed finite input energy per channel use.
\end{theorem}
At $n=0$, the theorem recovers the pure-loss capacity, which is already attained by LOCC~\cite{r11}:
\begin{equation}Q_{\mathrm{PPT}}(\mathcal L_{\eta,0})=Q_2(\mathcal L_{\eta,0})=-\log_2(1-\eta).\label{eq:pureloss}\end{equation}
At positive thermal noise below the entanglement-breaking threshold, it proves that the PLOB upper bound~\cite{r11} is attainable under PPT assistance.

To compare with LOCC, define the reverse-coherent-information benchmark $I_{\mathrm{RC}}$~\cite{PGPBL} and write $[z]_+=\max\{0,z\}$:
\begin{align}
 I_{\mathrm{RC}}(\eta,n)&=-\log_2(1-\eta)-g(n),\label{eq:rci}\\
 [I_{\mathrm{RC}}(\eta,n)]_+&\leq Q_2(\mathcal L_{\eta,n})\leq Q_{\mathrm{PPT}}(\mathcal L_{\eta,n}).\label{eq:loccbounds}
\end{align}
When $I_{\mathrm{RC}}\geq0$, the improvement over this benchmark is $n\log_2(1/\eta)$. This does not show that LOCC is limited to the benchmark. Positive LOCC rates extend beyond its positive region~\cite{r48}, and improved lower bounds are available~\cite{r14,r13}.

\paragraph{Thermal amplification.}
Let $\kappa>1$ be the amplifier gain and $n\geq0$ the mean occupation of an independent environment mode $e$ in $\tau_n$. The thermal amplifier $\mathcal A_{\kappa,n}$ acts as
\begin{equation}b=\sqrt\kappa\,a_{\mathrm{in}}+\sqrt{\kappa-1}\,e^\dagger.\label{eq:ampchannel}\end{equation}
It is entanglement breaking for $n\geq1/(\kappa-1)$~\cite{r32}. As shown in Appendix~\ref{app:G}, interchanging the two thermal modes in the thermal-loss construction attains the corresponding PLOB upper bound~\cite{r11}.
\begin{theorem}[Thermal-amplifier capacity]\label{thm:amp}
The PPT-assisted quantum capacity is
\begin{equation}
Q_{\mathrm{PPT}}(\mathcal A_{\kappa,n})=\begin{cases}
\log_2\dfrac{\kappa^{n+1}}{\kappa-1}-g(n),&n<\dfrac{1}{\kappa-1},\\[4pt]
0,&n\geq\dfrac{1}{\kappa-1}.
\end{cases}\label{eq:amplifier-capacity}
\end{equation}
Every strictly smaller rate is achievable with fixed finite input energy per use.
\end{theorem}
At $n=0$, the environment is in the vacuum and the channel is a quantum-limited amplifier. The theorem recovers its known LOCC capacity~\cite{r11}:
\begin{equation}
Q_{\mathrm{PPT}}(\mathcal A_{\kappa,0})=Q_2(\mathcal A_{\kappa,0})=-\log_2(1-\kappa^{-1}).\label{eq:quantumlimitedamp}
\end{equation}

\paragraph{Additive Gaussian noise.}
Let $\xi>0$ be the mean photon number added by random displacements. For a mode with annihilation operator $a$, write $D(\alpha)=\exp(\alpha a^\dagger-\bar\alpha a)$. The additive Gaussian-noise channel acts on a density operator $\omega$ as
\begin{equation}
\mathcal Z_\xi(\omega)=\int_{\mathbb C}\frac{d^2\alpha}{\pi\xi}\,e^{-|\alpha|^2/\xi}D(\alpha)\omega D(\alpha)^\dagger.\label{eq:additivechannel}
\end{equation}
In quadratures whose vacuum covariance is $I_2$, it adds $2\xi I_2$ to the input covariance. It is entanglement breaking for $\xi\geq1$~\cite{r32}. As shown in Appendix~\ref{app:H}, we attain the additive-noise PLOB upper bound~\cite{r11}.
\begin{theorem}[Additive-noise capacity]\label{thm:add}
The PPT-assisted quantum capacity is
\begin{equation}
Q_{\mathrm{PPT}}(\mathcal Z_\xi)=\begin{cases}
(\xi-1)/\ln2-\log_2\xi,&0<\xi<1,\\
0,&\xi\geq1.
\end{cases}\label{eq:additive-capacity}
\end{equation}
Every strictly smaller finite rate is achievable with fixed finite input energy. The noiseless endpoint $\mathcal Z_0=\mathrm{id}$ has infinite unconstrained capacity.
\end{theorem}

\paragraph{Entanglement distribution and secret key.}
Let us denote by $D_{\mathrm{PPT}}(\mathcal N)$ the channel's PPT-assisted entanglement-distribution capacity, in ebits per use. The protocol starts in a separable state, and its output on $K_m$-dimensional systems held by Alice and Bob must overlap with $\Phi_{K_m}=\proj{\phi_{K_m}}$, where $\ket{\phi_{K_m}}=K_m^{-1/2}\sum_{j=1}^{K_m}\ket{j}_A\ket{j}_B$ in the number basis $\ket{j}$. A transmission code applied to half of a locally prepared $\Phi_{K_m}$ distributes entanglement. Conversely, teleporting with distributed ebits gives a transmission code. As discussed in Appendix~\ref{app:A}, these conversions imply
\begin{equation}D_{\mathrm{PPT}}(\mathcal N)=Q_{\mathrm{PPT}}(\mathcal N).\label{eq:distribution}\end{equation}
Our converses bound $D_{\mathrm{PPT}}$, and hence $Q_{\mathrm{PPT}}$. Our codes distribute entanglement, and teleportation converts them into transmission codes. 


Let $K(\mathcal N)$ be the ordinary secret-key capacity with LOCC and public classical communication, and let $B_{\mathrm{PLOB}}(\mathcal N)$ denote the corresponding piecewise expression in Theorems~\ref{thm:loss}--\ref{thm:add}. For these three channel families without an energy constraint, PLOB's key converse~\cite{r11} and our results give
\begin{equation}
Q_2(\mathcal N)\leq K(\mathcal N)\leq B_{\mathrm{PLOB}}(\mathcal N)=Q_{\mathrm{PPT}}(\mathcal N).\label{eq:keybounds}
\end{equation}
Thus a rate-preserving reduction of an optimal PPT protocol to LOCC would also determine $K$. Lower bounds specifically on secret-key capacity~\cite{PGPBL,OLCSPB,OPB} need not be achievable entanglement-distribution rates.

As shown in Appendix~\ref{app:A}, the direct analogue $K_{\mathrm{PPT}}$, allowing the unrestricted PPT assistance used here, is infinite even without using the channel. A free PPT replacement map can prepare arbitrarily many copies of a PPT state with positive LOCC-distillable key~\cite{r38}. This remains true when Eve holds a purification and at zero transmitted energy. Thus $K_{\mathrm{PPT}}$ does not provide a finite comparison with ordinary $K$ in this assistance model.

\section{Energy-constrained extension}\label{sec:energy}
We now determine the capacity at each prescribed finite input energy. Thermal loss remains the principal example; the same function also treats thermal amplification and additive noise.

\paragraph{The transmitted-energy constraint.}
Let $N\geq0$ be the permitted mean number of transmitted photons per use. In an $r$-use adaptive protocol, let $\rho_{A_i}^{(i)}$ be the unconditional state of the transmitted mode immediately before use $i$, and let $\hat n_{A_i}$ be its photon-number operator. The energy-constrained capacity uses the same fidelity criterion and free operations as $Q_{\mathrm{PPT}}$, subject to
\begin{equation}
x_i=\operatorname{Tr}(\hat n_{A_i}\rho_{A_i}^{(i)}),~~\sum_{i=1}^{r}x_i\leq rN\label{eq:energy_constraint}
\end{equation}
for every input state. The average includes all classical branches. No bound is imposed on the energy of local memories. The constraint is on the transmitted mean, rather than on every photon-number outcome. As shown in Appendix~\ref{app:A}, the same operational conversions give $D_{\mathrm{PPT}}(\mathcal N,N)=Q_{\mathrm{PPT}}(\mathcal N,N)$ without additional channel uses or transmitted photons.

\paragraph{Exact energy-constrained thermal-loss capacity.}
Let $y$ denote the output mean and $\ell_0=1+N+y$. The correlation parameter $t$ and nonnegative thermal occupations $u,w$ in the squeezed-thermal representation are
\begin{equation}
\begin{aligned}
y&=\eta N+(1-\eta)n,~~\ell_0=1+N+y,\\
t&=\frac{\ell_0-\sqrt{\ell_0^2-4\eta N(N+1)}}{2}:=t(\eta),\\
u&=N-t,~~ w=y-t.
\end{aligned}\label{eq:loss_parameters}
\end{equation}
Define $\varphi(x)=x\log_2x$ for $x>0$ and $\varphi(0)=0$. Below the entanglement-breaking threshold, set
\begin{equation}
\begin{split}
F_{\eta,n}(N)={}&\varphi(t)+\varphi(\ell_0-t)-\varphi(N)-\varphi(y)\\
&-g(u)-g(w).
\end{split}\label{eq:energy_function}
\end{equation}
For $n\geq\eta/(1-\eta)$, set $F_{\eta,n}(N)=0$. This is the energy-constrained converse function of Ref.~\cite{r16}. It is continuous, nondecreasing and concave in $N$. As shown in Appendices~\ref{app:K}--\ref{app:N}, it is attained under PPT assistance.
\begin{theorem}[Energy-constrained thermal loss]\label{thm:energyloss}
For every $0<\eta<1$, $n\geq0$ and finite $N\geq0$,
\begin{equation}Q_{\mathrm{PPT}}(\mathcal L_{\eta,n},N)=F_{\eta,n}(N).\label{eq:energyloss}\end{equation}
Every smaller positive rate is attained with $N$ fixed throughout the coding limit. The capacity vanishes in the entanglement-breaking region.
\end{theorem}
Let $Q_2(\mathcal N,N)$ denote the same energy-constrained capacity with LOCC assistance. At zero thermal noise, the formula reduces to the known LOCC value~\cite{r15}:
\begin{equation}
\begin{split}
Q_{\mathrm{PPT}}(\mathcal L_{\eta,0},N)&=Q_2(\mathcal L_{\eta,0},N)\\
&=g(N)-g((1-\eta)N).
\end{split}\label{eq:energypureloss}
\end{equation}
For noisy thermal loss, equality with the optimal LOCC rate is not established.

\paragraph{Thermal amplification and additive noise.}
To state the extension uniformly, let $\tau>0$ be the intensity gain and $v\geq\max\{0,\tau-1\}$ the output mean photon number for vacuum input. Denote the resulting phase-insensitive Gaussian channel by $\mathcal G_{\tau,v}$. In quadratures with vacuum covariance $I_2$, its first-moment vector $\bar r$ and covariance matrix $V$ transform as
\begin{equation}
\bar r\longmapsto\sqrt\tau\,\bar r,~~ V\longmapsto\tau V+(2v+1-\tau)I_2.\label{eq:gaussianchannel}
\end{equation}
The three families correspond to
\begin{equation}
\begin{array}{c|cc}
\text{Channel}&\tau&v\\\hline
\mathcal L_{\eta,n}&\eta&(1-\eta)n\\
\mathcal A_{\kappa,n}&\kappa&(\kappa-1)(n+1)\\
\mathcal Z_\xi&1&\xi
\end{array}\label{eq:channelparameters}
\end{equation}
The output mean, correlation parameter and thermal occupations are now
\begin{equation}
\begin{aligned}
y&=\tau N+v,~~\ell_0=1+N+y,\\
t&=t(\tau),~~u=N-t,~~ w=y-t.
\end{aligned}\label{eq:general_parameters}
\end{equation}
For $v<\tau$, define $F_{\tau,v}(N)$ by the right-hand side of Eq.~\eqref{eq:energy_function} with the parameters in Eq.~\eqref{eq:general_parameters}. Set $F_{\tau,v}(N)=0$ for $v\geq\tau$. In particular, $F_{\eta,(1-\eta)n}=F_{\eta,n}$.

The extension also determines the entanglement of the channel output states. Let $|\psi_N\rangle$ be a two-mode squeezed vacuum with mean $N$ photons in each mode, and set $\chi_N^{\mathcal N}=(\mathrm{id}\otimes\mathcal N)(|\psi_N\rangle\langle\psi_N|)$. For a bipartite state $\rho$, $D_{\mathrm{PPT}}(\rho)$ is the asymptotic number of ebits distilled per copy under PPT assistance. The relative entropy of entanglement (REE), $E_R(\rho)$, is the infimum of the quantum relative entropy from $\rho$ to a separable state~\cite{r33,r34,r35,Simeone2026}. The Rains bound $R(\rho)$ minimizes the extended quantum relative entropy $D(\rho\Vert\sigma)$ over positive trace-class $\sigma$ with $\|\sigma^\Gamma\|_1\leq1$~\cite{r19,r20}. Appendix~\ref{app:A} defines the extended relative entropy without subtracting infinite entropies. For $E=R$ or $E=E_R$, define $E^\infty(\rho)=\inf_{\ell\geq1}E(\rho^{\otimes\ell})/\ell$. As shown in Appendices~\ref{app:K}--\ref{app:N}, the sector bounds, adaptive converse and distillation construction give the following equalities.
\begin{theorem}[Energy-constrained Gaussian channels]\label{thm:energygeneral}
For every physical $\mathcal G_{\tau,v}$ above and finite $N\geq0$, let $\chi_N=\chi_N^{\mathcal G_{\tau,v}}$. Then
\begin{equation}
\begin{aligned}
Q_{\mathrm{PPT}}(\mathcal G_{\tau,v},N)&=D_{\mathrm{PPT}}(\chi_N)=R^\infty(\chi_N)\\
&=E_R^\infty(\chi_N)=F_{\tau,v}(N).
\end{aligned}\label{eq:energygeneral}
\end{equation}
\end{theorem}
Thus the thermal-amplifier capacity is $F_{\kappa,(\kappa-1)(n+1)}(N)$ and the additive-noise capacity is $F_{1,\xi}(N)$. They vanish respectively for $n\geq1/(\kappa-1)$ and $\xi\geq1$. The state equalities do not require single-copy equality of $R$ and $E_R$. At the quantum-limited amplifier boundary, the common PPT and LOCC rate~\cite{r15} is
\begin{equation}
\begin{split}
&Q_{\mathrm{PPT}}(\mathcal A_{\kappa,0},N)=Q_2(\mathcal A_{\kappa,0},N)\\
&=g(\kappa N+\kappa-1)-g((\kappa-1)(N+1)).
\end{split}\label{eq:energyamp}
\end{equation}
At the noiseless additive endpoint, $Q_{\mathrm{PPT}}(\mathcal Z_0,N)=Q_2(\mathcal Z_0,N)=g(N)$. This is finite at fixed $N$, whereas the unrestricted capacity is infinite. For all three noisy families, taking $N\to\infty$ after the coding limit recovers the capacities in Sec.~\ref{sec:unconstrained}. As shown in Appendix~\ref{app:N}, the construction also handles hard total-photon budgets and channel boundaries. Appendix~\ref{app:O} compares alternative constructions. The Rains converse concerns entanglement distribution and hence quantum transmission. It does not by itself upper-bound the ordinary energy-constrained secret-key capacity.

\section{Conclusions}\label{sec:conclusions}
We have determined the PPT-assisted capacities of thermal-loss, thermal-amplifier and additive Gaussian-noise channels, both without an energy constraint and at each prescribed mean input energy. This solves a nontrivial noisy bosonic capacity problem by matching constructive protocols with converses for arbitrary adaptive strategies. The energy-constrained extension also identifies the regularized entanglement of the corresponding Gaussian channel states.

The remaining search for $Q_2=Q_{\mathrm{PPT}}$ requires a capacity-achieving LOCC protocol that preserves the rate and the imposed input-energy budget. In the unconstrained case, such a construction would also establish $K=Q_2=Q_{\mathrm{PPT}}$; unrestricted PPT-assisted key generation itself is trivial. Whether LOCC can attain these rates for noisy channels remains open. A bound valid under arbitrary PPT assistance cannot improve on the unconstrained PLOB value.

\paragraph{Acknowledgments.}
This work was supported by UK Research and Innovation (UKRI) through the Integrated Quantum Networks Research Hub (IQN, EP/Z533208/1). OpenAI's ChatGPT (GPT-6 Astra model), a large language model, was used during the preparation of this manuscript to generate candidate mathematical derivations and proofs, to assist with their development and to improve the presentation of the manuscript. Anthropic's Claude (Opus 5.5) was also used to check derivations numerically, verify references and literature, and assist with preparing the manuscript. Both models were directed by prompts written by the authors concerning the mathematical assumptions, derivations and statements under consideration. All mathematical results, derivations and proofs produced with their assistance were independently examined, verified and revised as necessary by the authors, who take full responsibility for the accuracy and scientific content of the work.

\clearpage
\onecolumn
\setcounter{page}{1}
\renewcommand{\thepage}{S\arabic{page}}
\begin{center}
{\large\bfseries Supplemental Material}\par\vspace{7pt}
{\normalsize\bfseries PPT-assisted capacities of bosonic Gaussian channels}\par\vspace{4pt}
Matthew Barber and Stefano Pirandola
\end{center}
\section*{Proof architecture}
Appendix~\ref{app:A} fixes the operations, states, entanglement measures and operational tasks. The diagrams organize the technical proofs by the statements each appendix supplies. Solid arrows show dependencies of the capacity proofs. Dashed boxes record supplementary consequences or comparisons, which are not premises of the main theorems.

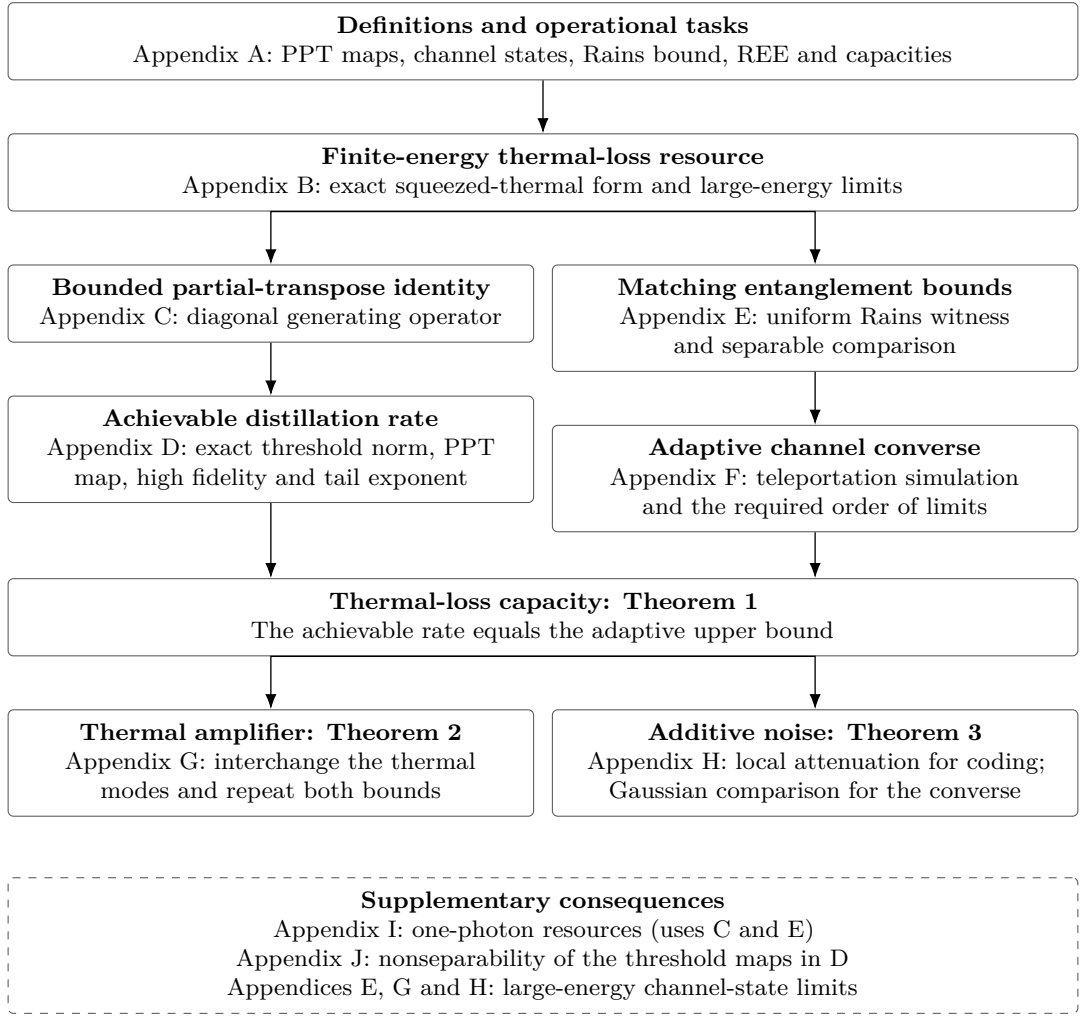
\begin{figure}[h!]
\centering
\begin{tikzpicture}[
 node distance=7mm and 10mm,
 box/.style={draw=black!65,rounded corners=2pt,align=center,font=\small,inner sep=5pt,text width=6.6cm},
 wide/.style={box,text width=13.8cm},
 aux/.style={wide,dashed},
 arr/.style={-{Latex[length=2mm]},line width=0.5pt}]
\node[wide] (a) {\textbf{Definitions and operational tasks}\\Appendix A: PPT maps, channel states, Rains bound, REE and capacities};
\node[wide,below=of a] (b) {\textbf{Finite-energy thermal-loss resource}\\Appendix B: exact squeezed-thermal form and large-energy limits};
\node[box,below=of b.south west,anchor=north west] (c) {\textbf{Bounded partial-transpose identity}\\Appendix C: diagonal generating operator};
\node[box,below=of b.south east,anchor=north east] (e) {\textbf{Matching entanglement bounds}\\Appendix E: uniform Rains witness\\and separable comparison};
\node[box,below=of c] (d) {\textbf{Achievable distillation rate}\\Appendix D: exact threshold norm, \mbox{PPT} map, high fidelity and tail exponent};
\node[box,below=of e] (f) {\textbf{Adaptive channel converse}\\Appendix F: teleportation simulation\\and the required order of limits};
\coordinate (mid) at ($(d.south)!0.5!(f.south)$);
\node[wide,below=8mm of mid] (t1) {\textbf{Thermal-loss capacity: Theorem 1}\\The achievable rate equals the adaptive upper bound};
\node[box,below=of t1.south west,anchor=north west] (g) {\textbf{Thermal amplifier: Theorem 2}\\Appendix G: interchange the thermal\\modes and repeat both bounds};
\node[box,below=of t1.south east,anchor=north east] (h) {\textbf{Additive noise: Theorem 3}\\Appendix H: local attenuation for coding;\\Gaussian comparison for the converse};
\coordinate (mid2) at ($(g.south)!0.5!(h.south)$);
\node[aux,below=8mm of mid2] (extra) {\textbf{Supplementary consequences}\\Appendix I: one-photon resources (uses C and E)\\Appendix J: nonseparability of the threshold maps in D\\Appendices E, G and H: large-energy channel-state limits};
\draw[arr] (a)--(b);\draw[arr] (b.south)-|(c.north);\draw[arr] (b.south)-|(e.north);
\draw[arr] (c)--(d);\draw[arr] (e)--(f);
\draw[arr] (d.south)--(d.south|-t1.north);\draw[arr] (f.south)--(f.south|-t1.north);
\draw[arr] (t1.south)-|(g.north);\draw[arr] (t1.south)-|(h.north);
\end{tikzpicture}
\caption{Unconstrained capacities. The left branch constructs a PPT code; the right branch bounds every adaptive PPT protocol. Both use physical finite-energy resource states, but take the energy and coding limits in different orders. The exact threshold norm in Appendix D also gives improved finite-block tests.}
\label{fig:unconstrained}
\end{figure}
\clearpage
\section*{Energy-constrained proof architecture}
Here the input energy $N$ stays fixed as the block length grows. The photon-sector spectra, exact sector simulation and density limit are inputs from Refs.~\cite{r15,r16}. Appendix~\ref{app:K} explains these inputs and constructs matching PPT witnesses. The proof then separates into a correlation-threshold code and an adaptive converse. It reuses the standard-form map of Appendix~\ref{app:D} and the witness inequality of Appendix~\ref{app:E}.

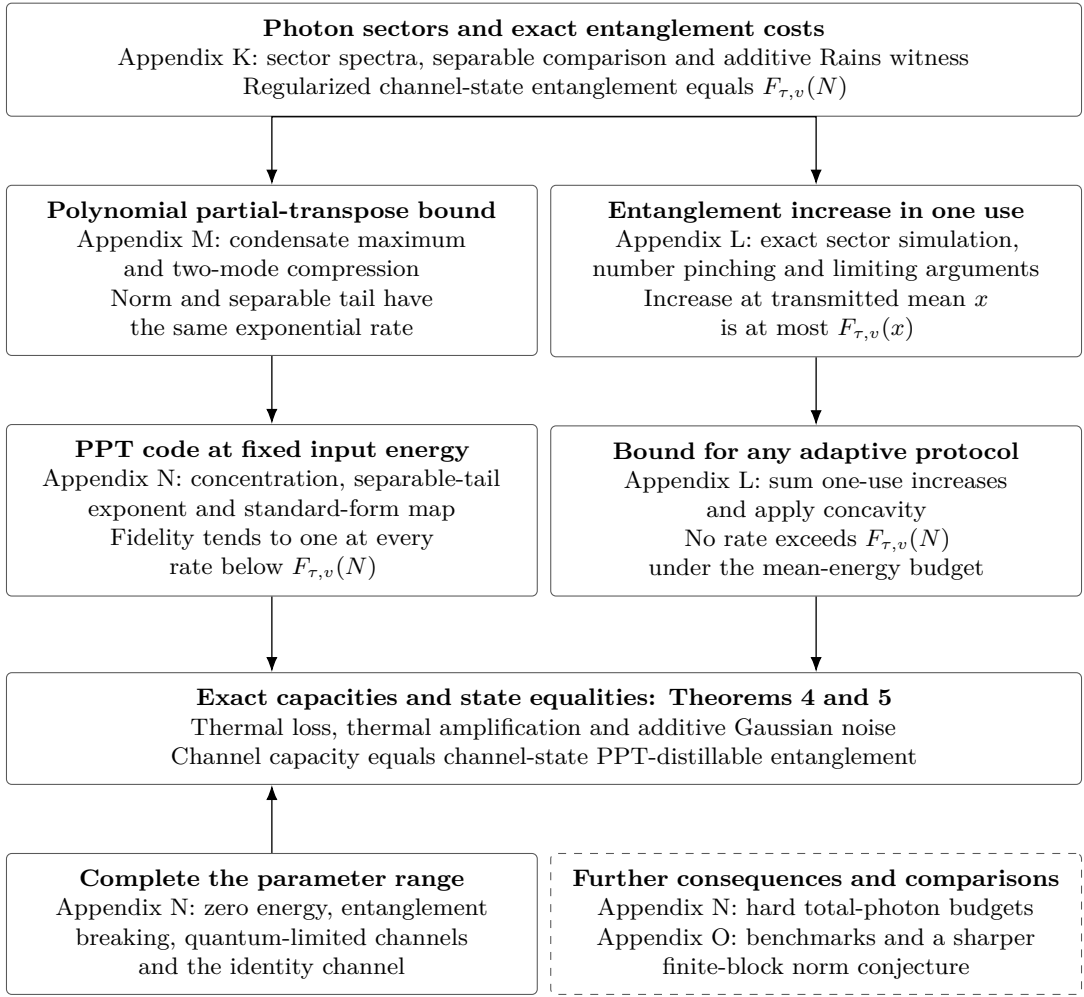
\begin{figure}[h!]
\centering
\begin{tikzpicture}[
 node distance=9mm and 10mm,
 box/.style={draw=black!65,rounded corners=2pt,align=center,font=\small,inner sep=6pt,text width=6.6cm},
 wide/.style={box,text width=13.8cm},
 aux/.style={box,dashed},
 arr/.style={-{Latex[length=2mm]},line width=0.5pt}]
\node[wide] (k) {\textbf{Photon sectors and exact entanglement costs}\\Appendix K: sector spectra, separable comparison and additive Rains witness\\Regularized channel-state entanglement equals $F_{\tau,v}(N)$};
\node[box,below=of k.south west,anchor=north west] (m) {\textbf{Polynomial partial-transpose bound}\\Appendix M: condensate maximum\\and two-mode compression\\Norm and separable tail have\\the same exponential rate};
\node[box,below=of k.south east,anchor=north east] (l1) {\textbf{Entanglement increase in one use}\\Appendix L: exact sector simulation,\\number pinching and limiting arguments\\Increase at transmitted mean $x$\\is at most $F_{\tau,v}(x)$};
\node[box,below=of m] (n) {\textbf{PPT code at fixed input energy}\\Appendix N: concentration, separable-tail\\exponent and standard-form map\\Fidelity tends to one at every\\rate below $F_{\tau,v}(N)$};
\node[box,below=of l1] (l2) {\textbf{Bound for any adaptive protocol}\\Appendix L: sum one-use increases\\and apply concavity\\No rate exceeds $F_{\tau,v}(N)$\\under the mean-energy budget};
\coordinate (mid) at ($(n.south)!0.5!(l2.south)$);
\node[wide,below=10mm of mid] (t) {\textbf{Exact capacities and state equalities: Theorems 4 and 5}\\Thermal loss, thermal amplification and additive Gaussian noise\\Channel capacity equals channel-state PPT-distillable entanglement};
\node[box,below=of t.south west,anchor=north west] (b) {\textbf{Complete the parameter range}\\Appendix N: zero energy, entanglement\\breaking, quantum-limited channels\\and the identity channel};
\node[aux,below=of t.south east,anchor=north east] (x) {\textbf{Further consequences and comparisons}\\Appendix N: hard total-photon budgets\\Appendix O: benchmarks and a sharper\\finite-block norm conjecture};
\draw[arr] (k.south)-|(m.north);\draw[arr] (k.south)-|(l1.north);
\draw[arr] (m)--(n);\draw[arr] (l1)--(l2);
\draw[arr] (n.south)--(n.south|-t.north);\draw[arr] (l2.south)--(l2.south|-t.north);
\draw[arr] (b.north)--(b.north|-t.south);
\end{tikzpicture}
\caption{Energy-constrained extension. The polynomial factor in the norm bound disappears after taking its logarithm and dividing by block length. The sharp norm conjecture in Appendix O is not needed. Appendix N completes the parameter range and separately establishes the hard-budget extension.}
\label{fig:energy}
\end{figure}
\clearpage
\twocolumn
\appendix
\renewcommand{\thesection}{\Alph{section}}
\titleformat{\section}{\normalfont\small\bfseries\centering}{Appendix \thesection:}{0.5em}{}
\numberwithin{equation}{section}
\renewcommand{\theequation}{\thesection\arabic{equation}}

\section{Preliminary notions}\label{app:A}
This appendix fixes the conventions and operational definitions used in the proofs.

\subsection{PPT operations and operator norms}
The identity channel is denoted by $\operatorname{id}$. We write $\Gamma_B$ for transposition of Bob's indices in the number basis, or $\Gamma$ when the system is clear. For a bipartite operator $X$,
\begin{equation}\label{eq:A1}
 \langle i,j|X^\Gamma|k,l\rangle=\langle i,l|X|k,j\rangle.
\end{equation}
A normalized state is PPT when $X^\Gamma\geq0$. For input systems $A,B$, output systems $A',B'$ and a completely positive trace-preserving map $\Lambda$, define
\begin{equation}\label{eq:A2}
 \Lambda^\Gamma=\Gamma_{B'}\circ\Lambda\circ\Gamma_B.
\end{equation}
We call $\Lambda$ completely PPT preserving when $\Lambda^\Gamma$ is also completely positive and trace preserving. Complete positivity includes arbitrary untouched ancillary systems, as required when a protocol retains quantum memories between channel uses.

On bosonic spaces both maps must be normal: their dual maps preserve increasing bounded limits of positive observables. Partial transpose need not be bounded on the full infinite-dimensional operator space, so identities involving it are first established on finite number-state matrix elements and then supplied with the required bounded or trace-class extension. Composition, tensor products and local classical control preserve the assistance class. LOCC is contained in this class, but a PPT map need not admit an LOCC implementation.

For an operator $X$, we use the trace norm $\|X\|_1=\operatorname{Tr}\sqrt{X^\dagger X}$, the operator norm $\|X\|_\infty=\sup_{\|\psi\|=1}\|X\psi\|$, and, when finite, the Hilbert--Schmidt norm $\|X\|_2=(\operatorname{Tr}X^\dagger X)^{1/2}$. The duality bound $|\operatorname{Tr}XY|\leq\|X\|_1\|Y\|_\infty$ connects Rains comparisons to PPT witnesses. Trace distance between states controls every measurement probability; in particular, for an effect $0\leq Q\leq I$, $|\operatorname{Tr}Q(\rho-\sigma)|\leq\|\rho-\sigma\|_1/2$.

\subsection{Quantum communication and rates}
For a channel $\mathcal N$ from Alice to Bob, $Q_{\mathrm{PPT}}(\mathcal N)$ denotes the quantum capacity with this assistance. A code of length $m$ and dimension $K_m$ acts on a $K_m$-dimensional system $A_0$ held by Alice. Alice and Bob share no entanglement beforehand: any ancillary states start in a separable state that is independent of the input. The protocol interleaves $m$ uses of $\mathcal N$ with completely PPT-preserving operations. It therefore defines a channel $\Lambda_m$ from $A_0$ to a $K_m$-dimensional system $B_0$ held by Bob, identified with $A_0$ through the number basis. Let $\widetilde A$ be a reference system of the same dimension, which no operation acts on. The code's worst-case fidelity is
\begin{equation}\label{eq:A3}
 f_m=\min_\Psi\langle\Psi|(\operatorname{id}_{\widetilde A}\otimes\Lambda_m)(|\Psi\rangle\langle\Psi|)|\Psi\rangle,
\end{equation}
where the minimum is over all unit vectors $\Psi$ in $\widetilde A\otimes A_0$. A rate $R$ is achievable if, for each $m$, there is a code of length $m$ and dimension $K_m$, with worst-case fidelity $f_m$, such that $f_m\to1$ and $\liminf_{m\to\infty}m^{-1}\log_2K_m\geq R$. The supremum of these rates is $Q_{\mathrm{PPT}}(\mathcal N)$, in qubits per channel use. For an integer $K\geq2$, let $|\phi_K\rangle$ be the maximally entangled vector on two $K$-dimensional systems, with number-basis labels $j$, and let $\Phi_K$ be its density operator:
\begin{equation}\label{eq:A4}
 |\phi_K\rangle=K^{-1/2}\sum_{j=1}^K|j\rangle_A|j\rangle_B,
 ~~ \Phi_K=|\phi_K\rangle\langle\phi_K|.
\end{equation}
One ebit is one maximally entangled qubit pair. The capacity $Q_2$ uses LOCC assistance instead.

\subsection{Bosonic and Gaussian conventions}
For each mode, $[a,a^\dagger]=1$ and the photon-number operator is $a^\dagger a$. We use $q=a+a^\dagger$ and $p=(a-a^\dagger)/i$, so $[q,p]=2i$ and the vacuum covariance is $I_2$. For a quadrature vector $r$, let $\Delta r_j=r_j-\langle r_j\rangle$. The covariance entries are
\begin{equation}\label{eq:A5}
 V_{jk}=\tfrac12\langle\Delta r_j\Delta r_k+\Delta r_k\Delta r_j\rangle.
\end{equation}
A Gaussian state is determined by its first moments and covariance~\cite{r2,r3}. The states used below are centered. For two modes, transposition of Bob reverses the sign of $p_B$ in the covariance matrix. Gaussian PPT and separability criteria are developed in Refs.~\cite{r26,r27,r28}.

For mean occupation $x\geq0$, the thermal state is
\begin{equation}\label{eq:A6}
 \tau_x=\frac1{1+x}\sum_{j=0}^\infty\left(\frac{x}{1+x}\right)^j|j\rangle\langle j|.
\end{equation}
Its entropy is $g(x)$ from Eq.~\eqref{eq:entropy}, and its number variance is $x(1+x)$. Gaussian operations alone cannot distill Gaussian entanglement~\cite{r29,r30,r31}; the collective PPT maps in this paper are not restricted to Gaussian operations.

A channel is entanglement breaking if applying it to one part of any state gives a separable state. Here separable states are trace-norm limits of mixtures of product states.

\paragraph{Thermal loss.} Let $\eta\in(0,1)$ be the transmissivity and $n\geq0$ the environment's mean photon number. The thermal-loss channel $\mathcal L_{\eta,n}$ mixes an input annihilation operator $a_{\mathrm{in}}$ with an independent environment mode $e$ in $\tau_n$. The output annihilation operator is $b$ and the beam-splitter relation is
\begin{equation}\label{eq:A7}
 b=\sqrt\eta\,a_{\mathrm{in}}+\sqrt{1-\eta}\,e.
\end{equation}
The channel is entanglement breaking at and above $n=\eta/(1-\eta)$~\cite{r32}.

\paragraph{Thermal amplification.} Let $\kappa>1$ be the amplifier gain and $n\geq0$ the mean occupation of an independent environment mode $e$ in $\tau_n$. The thermal amplifier $\mathcal A_{\kappa,n}$ acts on the input mode $a_{\mathrm{in}}$, with output mode $b$, as
\begin{equation}\label{eq:A8}
 b=\sqrt\kappa\,a_{\mathrm{in}}+\sqrt{\kappa-1}\,e^\dagger.
\end{equation}
It is entanglement breaking for $n\geq1/(\kappa-1)$~\cite{r32}.

\paragraph{Additive Gaussian noise.} Let $\xi>0$ be the mean photon number added by random displacements. For a mode with annihilation operator $a$, write $D(\alpha)=\exp(\alpha a^\dagger-\bar\alpha a)$ for displacement by a complex amplitude $\alpha$. The additive Gaussian-noise channel $\mathcal Z_\xi$ acts on a density operator $\omega$ as
\begin{equation}\label{eq:A9}
 \mathcal Z_\xi(\omega)=\int_{\mathbb C}\frac{d^2\alpha}{\pi\xi}e^{-|\alpha|^2/\xi}D(\alpha)\omega D(\alpha)^\dagger.
\end{equation}
In quadratures whose vacuum covariance is the two-dimensional identity matrix $I_2$, it adds $2\xi I_2$ to the input covariance. It is entanglement breaking for $\xi\geq1$~\cite{r32}.

\subsection{Channel states and entanglement measures}
We give the explicit channel states and entanglement definitions used throughout the proofs. For a mean photon number $N\geq0$, let $|\psi_N\rangle$ be the two-mode squeezed vacuum, with number-basis labels $j$. For any one of the channels $\mathcal N$ above, let $\chi_N^{\mathcal N}$ be the state produced by sending its second mode through $\mathcal N$:
\begin{align}
 |\psi_N\rangle&=\sum_{j=0}^\infty\sqrt{\frac{N^j}{(N+1)^{j+1}}}|j\rangle|j\rangle,\label{eq:A10}\\
 \chi_N^{\mathcal N}&=(\operatorname{id}\otimes\mathcal N)(|\psi_N\rangle\langle\psi_N|).\label{eq:A11}
\end{align}
Each marginal of $|\psi_N\rangle$ has mean photon number $N$. At $N=0$, use the vacuum state $|\psi_0\rangle=|0\rangle|0\rangle$.

For a state $\rho$, let $D_{\mathrm{PPT}}(\rho)$ be its distillable entanglement under the same PPT operations. We use the extended Umegaki quantum relative entropy, with logarithms to base two. For nonzero positive trace-class $\sigma$, its extension to an unnormalized second argument is
\begin{equation}\label{eq:A12}
 D(\rho\|\sigma)=D\left(\rho\left\|\frac{\sigma}{\operatorname{Tr}\sigma}\right.\right)-\log_2\operatorname{Tr}\sigma,
\end{equation}
and $D(\rho\|0)=+\infty$. This agrees with $\operatorname{Tr}\rho(\log_2\rho-\log_2\sigma)$ when that expression is well defined. It is infinite if the support condition fails. The extended definition avoids subtracting two infinite entropies, which matters for unrestricted quantum memories. The Rains bound $R$ minimizes over positive operators with partial-transpose trace norm at most one; the relative entropy of entanglement $E_R$ minimizes over separable states~\cite{r19,r33,r34,r35}:
\begin{align}
 R(\rho)&=\inf_{\substack{\sigma\geq0\\\|\sigma^\Gamma\|_1\leq1}}D(\rho\|\sigma),\label{eq:A13}\\
 E_R(\rho)&=\inf_{\sigma\text{ separable state}}D(\rho\|\sigma).\label{eq:A14}
\end{align}
Let $\ell\geq1$ be an integer number of copies. For either subadditive quantity $E=R$ or $E=E_R$, its regularization is $E^\infty(\rho)=\inf_{\ell\geq1}E(\rho^{\otimes\ell})/\ell$.

The minimization domain for $R$ contains every separable state, so $R\leq E_R$ and $R^\infty\leq E_R^\infty$. The operational bound $D_{\mathrm{PPT}}\leq R^\infty$ follows from the output test proved in Appendix~\ref{app:E}. A separable comparison therefore bounds the PPT rate through $R$; no PPT monotonicity of $E_R$ is required.

The Rains bound is monotone under these maps. If $\sigma$ is feasible in its definition, trace-norm contractivity for Hermitian inputs gives
\begin{equation}\label{eq:A15}
 \|\Lambda(\sigma)^\Gamma\|_1=\|\Lambda^\Gamma(\sigma^\Gamma)\|_1\leq\|\sigma^\Gamma\|_1\leq1.
\end{equation}
Relative-entropy data processing~\cite{r36} then proves monotonicity.

\subsection{Classical flags, entropy and energy budgets}
For a probability distribution $p_j$, write $H(J)=-\sum_jp_j\log_2p_j$ for its Shannon entropy. For a state $\rho$, write $S(\rho)=-\operatorname{Tr}\rho\log_2\rho$ for its von Neumann entropy. The binary entropy is $h_2(p)=-p\log_2p-(1-p)\log_2(1-p)$, with endpoint values understood by continuity.

A local flag is an orthogonal classical register recording a measurement outcome. For a flagged state $\sum_jp_j|j\rangle\langle j|\otimes\rho_j$, both $R$ and $E_R$ equal the average of the conditional values, when the register is local. Appendix~\ref{app:K} proves the identity in the setting needed here. This is distinct from appending a nonlocal test-outcome flag to a PPT distillation channel: the latter can destroy complete PPT preservation.

Number pinching removes coherence between different total photon counts. If $P_j$ is the projector onto count $j$, the pinching channel is $\mathcal P(\rho)=\sum_jP_j\rho P_j$. Its entropy cost is at most $H(J)$; the corresponding Rains bound is proved in Appendix~\ref{app:L}. A nonnegative integer count of mean $E$ has entropy at most $g(E)$. Thus a collective count of mean proportional to the block length $s$ costs only $O(\log s)$ bits to record, which is negligible per channel use.

The energy constraint applies to the unconditional transmitted means in Eq.~\eqref{eq:energy_constraint}. Local memories, ancillary systems and free operations have no separate energy budget. A hard total-photon constraint is stronger: the entire transmitted block must have support on counts at most $sN$.

\subsection{Entanglement distribution and quantum communication}\label{sec:A6}
For a state $\rho$, $D_{\mathrm{PPT}}(\rho)$ counts ebits per independent copy~\cite{r20}. For a channel $\mathcal N$, $D_{\mathrm{PPT}}(\mathcal N)$ counts ebits generated per use, starting from a separable state, with arbitrary adaptive PPT assistance. The output on $K_m$-dimensional systems held by Alice and Bob must have overlap with $\Phi_{K_m}$ tending to one. It equals $Q_{\mathrm{PPT}}(\mathcal N)$ of Sec.~II. The two conversions are as follows.

First, suppose a code has worst-case fidelity at least $1-\epsilon$. Alice applies it to one half of a locally prepared $\Phi_K$. The initial state lies entirely in Alice's laboratory, so this is a distribution protocol. Since $\Phi_K$ is among the states in the minimum, its output has overlap at least $1-\epsilon$ with $\Phi_K$. The reduced input state is maximally mixed, which is an average of pure input states. Since the transmitted mean is linear in the input, the protocol obeys any energy budget that the code obeys for every input.

Second, suppose a distribution protocol outputs $\omega$ with $\operatorname{Tr}(\Phi_K\omega)=f$. Alice and Bob first apply the local $U\otimes\bar U$ twirl, which preserves this overlap and gives the isotropic state
\[
 \bar\omega=f\Phi_K+(1-f)\frac{I-\Phi_K}{K^2-1}.
\]
They then use standard teleportation through $\bar\omega$~\cite{r37}. Each maximally entangled basis state teleports the input up to a Weyl operator, so linearity gives the channel $\rho\mapsto\sum_kp_kW_k\rho W_k^\dagger$, with $W_0=I$, $p_0=f$ and $p_k=(1-f)/(K^2-1)$ for $k\ne0$. Every pure state $\Psi$ of $\widetilde A\otimes A_0$ is therefore reproduced with fidelity
\[
 \sum_kp_k\bigl|\langle\Psi|(I\otimes W_k)|\Psi\rangle\bigr|^2\geq f,
\]
because the identity term alone contributes $f$.

Neither conversion adds a use of $\mathcal N$ or any transmitted photons. Hence
\begin{equation}\label{eq:A16}
 D_{\mathrm{PPT}}(\mathcal N,N)=Q_{\mathrm{PPT}}(\mathcal N,N),
\end{equation}
as well as Eq.~\eqref{eq:distribution} without a constraint.

\subsection{Secret key and the scope of PPT assistance}
Let $K_D(\rho)$ be the ordinary distillable secret key of a state under LOCC and public communication. Security requires the final key to be uniform, identical for Alice and Bob, and decoupled from Eve, including the public transcript. For key alphabet size $d$, the ideal classical--quantum state has the form
\begin{equation}\label{eq:A17}
 \gamma_{ABE}^{(d)}=\frac1d\sum_{j=1}^d|j\rangle\langle j|_A\otimes|j\rangle\langle j|_B\otimes\omega_E.
\end{equation}
The protocol output must approach such a state in trace norm; Eve may initially hold a purification of the shared resource. The capacity $K(\mathcal N)$ uses the same security criterion per channel use with ordinary LOCC assistance. Measuring distilled ebits yields secure bits, so $Q_2(\mathcal N)\leq K(\mathcal N)$.

Let $B_{\mathrm{PLOB}}(\mathcal N)$ denote the corresponding piecewise expression in Theorems~1--3. For these three channel families without an energy constraint, PLOB's key converse and our results give~\cite{r11}
\begin{equation}\label{eq:A18}
 Q_2(\mathcal N)\leq K(\mathcal N)\leq B_{\mathrm{PLOB}}(\mathcal N)=Q_{\mathrm{PPT}}(\mathcal N).
\end{equation}
An LOCC construction attaining the same rates, under the same input-energy budget, would establish $Q_2=Q_{\mathrm{PPT}}$. In the unconstrained case, such a construction would also establish $K=Q_2=Q_{\mathrm{PPT}}$; unrestricted PPT-assisted key generation itself is trivial.

PPT states with $K_D>0$ exist~\cite{r38}. This has a direct consequence for the formal quantity $K_{\mathrm{PPT}}$: here it means the same key task with the unrestricted completely PPT-preserving assistance used for $Q_{\mathrm{PPT}}$, including arbitrary local ancillary dimensions.

\begin{appprop}[Free key under unrestricted PPT assistance]\label{prop:A1}
For every channel and every transmitted-energy budget $N\geq0$, this definition gives
\begin{equation}\label{eq:A19}
 K_{\mathrm{PPT}}(\mathcal N)=K_{\mathrm{PPT}}(\mathcal N,N)=+\infty.
\end{equation}
\end{appprop}
\begin{proof}
Choose a finite-dimensional PPT state $\sigma$ with $K_D(\sigma)>0$. The replacement map $\Lambda_\sigma(X)=\operatorname{Tr}(X)\sigma$ is a channel, and its partial-transpose conjugate is
\begin{equation}\label{eq:A20}
 \Lambda_\sigma^\Gamma(X)=\operatorname{Tr}(X)\sigma^\Gamma.
\end{equation}
Because $\sigma^\Gamma$ is positive with unit trace, this is also a channel. The same argument applies to $\sigma^{\otimes L}$ for every finite $L$. Thus the free assistance can prepare arbitrarily many copies of $\sigma$, after which an LOCC key-distillation protocol produces a positive asymptotic number of secret bits per copy, even against a purifying Eve.

For any proposed finite rate per channel use, choose the number of free copies proportional to the number of uses with a sufficiently large proportionality factor. The resulting key rate exceeds the proposal while the security error vanishes. The channel can be ignored, or used on vacuum inputs, so this argument also obeys a zero transmitted-energy budget.
\end{proof}
Ordinary $K$ uses LOCC assistance. Since $R(\sigma)=0$ for every PPT state, including states with $K_D(\sigma)>0$, the key bound in Eq.~\eqref{eq:A18} uses PLOB's secret-key converse.

\subsection{Parameters and orders of limits}
The channel noise $n$ and input mean $N$ are physical parameters. The letters $H,M,u$ describe the squeezed-thermal normal form, with squeezing gain $H$. For thermal loss, its pair $(M,u)$ equals the pair $(u,w)$ in Eq.~\eqref{eq:loss_parameters}. The correlation parameter $t$ in Eq.~\eqref{eq:loss_parameters} is real; a photon-count threshold $t$ is an integer. In the sector arguments, $s$ is the number of modes, $m,k$ are Alice's and Bob's total photon counts, and $h$ labels a collective correlation space.

For unrestricted achievability, first choose finite $N$ large enough for a target rate and then take the coding block to infinity. In the unrestricted converse, first fix the entire adaptive protocol and let the simulation resource energy tend to infinity; only afterwards does the protocol length grow. For the energy-constrained results, $N$ remains fixed throughout.

\section{Thermal-loss resource states}\label{app:B}
The unrestricted construction requires an exact normal form for the physical state $\chi_N^{\mathcal L_{\eta,n}}$ at finite input energy. We express it as a two-mode squeeze of independent thermal states. One thermal occupation diverges with $N$, whereas the other remains finite and determines the error-threshold exponent. This separation is what makes the coding and converse bounds meet.

The normal-mode transformation is a change of coordinates. The protocol is defined directly by its PPT distillation map.

\subsection{Exact squeezed-thermal form}
Let $S_H=\exp[r(a^\dagger b^\dagger-ab)]$, where $r>0$ and $H=\cosh^2r>1$. Thus $S_H^\dagger aS_H=\sqrt H\,a+\sqrt{H-1}\,b^\dagger$, with the analogous relation for $b$. For thermal occupations $M,u\geq0$, set
\begin{equation}\label{eq:B1}
 \rho_{H,M,u}=S_H(\tau_M\otimes\tau_u)S_H^\dagger,
 ~~ \theta=\frac{H-1}{H}.
\end{equation}
The pre-squeezer modes are denoted $c,d$. The error mode is $d$; its mean occupation is $u$.

For thermal loss, abbreviate $\chi_N^{\mathcal L_{\eta,n}}$ to $\chi_N$ and set $G=(1-\eta)^{-1}$. Its covariance has blocks $a_NI_2,b_NI_2,c_NZ$, where $Z=\operatorname{diag}(1,-1)$ and
\begin{align}
 a_N&=2N+1,~~ b_N=2\eta N+2(1-\eta)n+1,\label{eq:B2}\\
 c_N&=2\sqrt{\eta N(N+1)},\label{eq:B3}\\
 d_N&=\sqrt{(a_N+b_N)^2-4c_N^2}.\label{eq:B4}
\end{align}
\begin{applemma}[Exact diagonalization]\label{lem:B1}
For $N>n$, define
\begin{align}
 M_N&=\frac{d_N+a_N-b_N-2}{4},\label{eq:B5}\\
 u_N&=\frac{d_N-a_N+b_N-2}{4},\label{eq:B6}\\
 H_N&=\frac{a_N+b_N+d_N}{2d_N}.\label{eq:B7}
\end{align}
Then $\chi_N=\rho_{H_N,M_N,u_N}$, and
\begin{align}
 H_N&\longrightarrow G,~~ u_N\longrightarrow n,\label{eq:B8}\\
 M_N&=(1-\eta)N+\eta n+O(N^{-1}).\label{eq:B9}
\end{align}
\end{applemma}
\begin{proof}
The covariance of Eq.~\eqref{eq:B1} has scalar blocks
\begin{align}
 a&=H(2M+1)+(H-1)(2u+1),\label{eq:B10}\\
 b&=(H-1)(2M+1)+H(2u+1),\label{eq:B11}\\
 c&=2\sqrt{H(H-1)}(M+u+1).\label{eq:B12}
\end{align}
Solving these equations gives Eq.~\eqref{eq:B7}. More explicitly, $a-b=2(M-u)$ and $\sqrt{(a+b)^2-4c^2}=2(M+u+1)$. These two invariants first determine the thermal occupations; $a+b=2(2H-1)(M+u+1)$ then determines the squeezing gain. The parameters are admissible because $a_N-b_N=2(1-\eta)(N-n)>0$ and
\[
 d_N^2-(a_N-b_N+2)^2=16(1-\eta)n(N+1)\geq0.
\]
Thus $u_N\geq0$ and $M_N>0$; moreover, $c_N>0$ gives $a_N+b_N>d_N$ and hence $H_N>1$. Both states are centred and Gaussian, so their covariances establish equality. Finally,
\begin{equation}\label{eq:B13}
 d_N=2(1-\eta)N+2(1+\eta)n+2+O(N^{-1}).
\end{equation}
Substitution gives Eq.~\eqref{eq:B9}.
\end{proof}

\section{Exact partial-transpose identity}\label{app:C}
This appendix proves the generating identity that controls both the threshold test and the Rains witness. Expanding $A_H(z)$ in powers of $z$ separates the projectors onto different error-photon counts. The identity below turns their partial transposes into functions of a single passive mode. Coefficient estimates will control the collective coding test, while positive real values of $z$ will provide converse witnesses.

For $|z|<1$ define the bounded operator
\begin{equation}\label{eq:C1}
 A_H(z)=S_H(I_c\otimes z^{d^\dagger d})S_H^\dagger,
\end{equation}
where $0^{d^\dagger d}=|0\rangle\langle0|$. Let $\|\cdot\|_\infty$ denote the operator norm.
\begin{applemma}[Generating identity]\label{lem:C1}
Define a passive mode $v$ and a scalar function $w$ by
\begin{equation}\label{eq:C2}
 v=\frac{\sqrt{H-1}\,a-\sqrt H\,b}{\sqrt{2H-1}},
 ~~ w(z)=\frac{z-\theta}{1-\theta z}.
\end{equation}
Then, as bounded operators,
\begin{equation}\label{eq:C3}
 A_H(z)^\Gamma=\frac{w(z)^{v^\dagger v}}{H(1-\theta z)}.
\end{equation}
The right-hand side is the identity on the orthogonal mode $(\sqrt H\,a+\sqrt{H-1}\,b)/\sqrt{2H-1}$.
\end{applemma}
\begin{proof}
The Bargmann representation sends $|j\rangle$ to $x^j/\sqrt{j!}$ and uses measure $e^{-|x|^2}d^2x/\pi$~\cite{r39}. Write $x,y$ for output variables and $p,q$ for conjugate input variables in a kernel. With $\lambda=\sqrt\theta$, the squeezer kernel is
\begin{equation}\label{eq:C4}
 K_{S_H}(x,y;p,q)=H^{-1/2}\exp\!\left[\lambda xy-\lambda pq+\frac{xp+yq}{\sqrt H}\right].
\end{equation}
It follows, for example, from the disentangling identity
\begin{equation}\label{eq:C5}
 S_H=e^{\lambda a^\dagger b^\dagger}H^{-(a^\dagger a+b^\dagger b+1)/2}e^{-\lambda ab}.
\end{equation}
Composing this kernel, the kernel of $I\otimes z^{d^\dagger d}$ and the adjoint kernel gives
\begin{equation}\label{eq:C6}
 K_{A_H(z)}=\frac{\exp[Cxp+zCyq+B(xy+pq)]}{H(1-\theta z)},
\end{equation}
where $C=(1-\theta)/(1-\theta z)$ and $B=\lambda(1-z)/(1-\theta z)$. The Gaussian integrations give the exponent
\begin{equation}\label{eq:C7}
 \lambda(xy+pq)+\frac{xp}{H}+\frac{z(y-\lambda p)(q-\lambda x)}{H(1-\theta z)}.
\end{equation}
Expanding it yields Eq.~\eqref{eq:C6}; the Gaussian determinant gives $(1-\theta z)^{-1}$.

Partial transposition exchanges $y$ and $q$, which carry Bob's output and input indices. The exponent then has coefficient matrix
\begin{equation}\label{eq:C8}
 \begin{pmatrix}C&B\\B&zC\end{pmatrix}.
\end{equation}
Its orthonormal eigenvectors are $(1,\sqrt\theta)/\sqrt{1+\theta}$ and $(\sqrt\theta,-1)/\sqrt{1+\theta}$, with eigenvalues $1,w(z)$. The second vector gives $v$. The Bargmann kernel of $s^{v^\dagger v}$ is the exponential of $s$ times the corresponding output and conjugate input coordinates. This proves Eq.~\eqref{eq:C3} on finite number-state matrix elements.

Finally,
\begin{equation}\label{eq:C9}
 1-|w(z)|^2=\frac{(1-\theta^2)(1-|z|^2)}{|1-\theta z|^2}\geq0.
\end{equation}
Thus the right-hand side is bounded and supplies the claimed extension. On compact subdisks, $|w(z)|$ is uniformly smaller than one, so it also defines an operator-norm analytic function there. Indeed, the spectral terms of $w(z)^{v^\dagger v}$ above occupation $J$ have total norm at most $\sup_{j>J}|w(z)|^j$, which tends uniformly to zero on each compact subdisk.
\end{proof}

\section{Threshold distillation}\label{app:D}
This appendix constructs the PPT protocol and proves its fidelity and rate, including a finite-block example. The two requirements are separate: the test must accept the resource with probability tending to one, and its partial transpose must have small operator norm. The first is a typicality statement for thermal photons. The second fixes the number of ebits that a completely PPT-preserving map may output.

For $L$ copies let $N_d=\sum_{i=1}^Ld_i^\dagger d_i$ be the total error-photon number. For an integer $t\geq0$, define the projection
\begin{equation}\label{eq:D1}
 Q_{L,t}=S_H^{\otimes L}(I_{c^L}\otimes\mathbf1_{\{N_d\leq t\}})(S_H^\dagger)^{\otimes L}.
\end{equation}
This is a mathematical test, not an assumed local measurement of $N_d$.

\subsection{Standard-form distillation}
\begin{applemma}[Rains's standard form~\cite{r20}]\label{lem:D1}
For $0\leq Q\leq I$ and an integer $K\geq2$, define
\begin{equation}\label{eq:D2}
 \mathcal R_{Q,K}(X)=\operatorname{Tr}(QX)\Phi_K+\operatorname{Tr}[(I-Q)X]\frac{I-\Phi_K}{K^2-1}.
\end{equation}
It is a channel, with fidelity $\operatorname{Tr}(Q\omega)$ on $\omega$. It is completely PPT preserving if and only if
\begin{equation}\label{eq:D3}
 -I/K\leq Q^\Gamma\leq I/K.
\end{equation}
\end{applemma}
\begin{proof}
The two effects are positive and sum to $I$, and the prepared states have unit trace. Hence the map is completely positive and trace preserving. For finite input dimension the Choi criterion characterizes complete positivity~\cite{r40}. Let $V_K$ swap the two output systems and let $\Pi_\pm=(I\pm V_K)/2$. Since $\Phi_K^\Gamma=V_K/K$, the Choi operator of the conjugated map has the two blocks
\begin{equation}\label{eq:D4}
 \frac{I/K+(Q^\Gamma)^T}{K+1}\otimes\Pi_+
 +\frac{I/K-(Q^\Gamma)^T}{K-1}\otimes\Pi_-.
\end{equation}
Here $T$ is a full input transpose. Positivity is equivalent to Eq.~\eqref{eq:D3}. For bosonic inputs, the extension is explicit. Put
\[
 A_\pm=\frac{I/K\pm Q^\Gamma}{K\pm1}.
\]
Then the conjugated map is
\[
 \mathcal R_{Q,K}^\Gamma(X)
 =\operatorname{Tr}(A_+X)\Pi_+
  +\operatorname{Tr}(A_-X)\Pi_-.
\]
Under Eq.~\eqref{eq:D3}, both coefficient operators are positive and bounded, so this map is normal and completely positive. Since $\operatorname{Tr}\Pi_\pm=K(K\pm1)/2$, their weighted sum is $I$, proving trace preservation. Conversely, complete positivity on every finite local input restriction gives Eq.~\eqref{eq:D3} on the dense finite number-state span, and hence the required bounded extension. The two prepared states in Eq.~\eqref{eq:D2} have orthogonal supports, proving the fidelity statement.
\end{proof}
The channel in Eq.~\eqref{eq:D2} has no test-outcome flag. Adding such a flag changes the map and need not preserve its PPT property.

\subsection{Exact norm, fidelity and asymptotic rate}
Put $a=H-1$ and $\theta=a/H$. Let $Z_{L,x}$ be the sum of $L$ independent geometric random variables of mean $x\geq0$, with
\begin{equation}\label{eq:D5}
 \Pr(Z_{L,x}=k)=\binom{L+k-1}{k}\frac{x^k}{(1+x)^{L+k}},~~ k\geq0.
\end{equation}
At $x=0$ this distribution is concentrated at zero. Write $p_{L,t}(x)=\Pr(Z_{L,x}\leq t)$ and set $p_{L,-1}(x)=0$. For $0\leq x<a$, define
\begin{equation}\label{eq:D6}
 F_H(x)=\log_2H-g(x)+x\log_2(1/\theta).
\end{equation}

\begin{appprop}[Exact threshold norm]\label{prop:D2}
For every $L\geq1$ and integer $0\leq t\leq a(L-1)$,
\begin{equation}\label{eq:D7}
 \|Q_{L,t}^\Gamma\|_\infty=p_{L,t}(a).
\end{equation}
In addition, for every $H>1$ and $L\geq1$,
\begin{equation}\label{eq:D8}
 \|Q_{L,1}^\Gamma\|_\infty=H^{-L}(1+L\theta).
\end{equation}
In particular, the norm at $t=0$ is $H^{-L}$.
\end{appprop}
\begin{proof}
Lemma~\ref{lem:C1} diagonalizes the partial transpose in the total occupation $j$ of the passive $v$ modes. Its eigenvalue is
\begin{equation}\label{eq:D9}
 \lambda_j=H^{-L}[z^t]\frac{w(z)^j}{(1-z)(1-\theta z)^L},~~ j\geq0.
\end{equation}
The coefficient identity first holds on finite number-state matrix elements. The bounded analytic identity of Lemma~\ref{lem:C1} then supplies its operator extension.

Set $T_k=\binom{L+k-1}{k}\theta^k$ for $k\geq0$, put $T_{-1}=0$, and let $h_k=\sum_{i=0}^kT_i$. Define $c_m=h_{t-m}$ for $0\leq m\leq t$, and $c_m=0$ for $m>t$. Since
\begin{equation}\label{eq:D10}
 \frac{T_k}{T_{k-1}}=\frac{\theta(L+k-1)}{k},
\end{equation}
our hypothesis implies $T_k\geq T_{k-1}$ for $1\leq k\leq t$. Hence $\Delta^2c_m=c_m-2c_{m+1}+c_{m+2}\geq0$ for every $m\geq0$, including the end of the sequence.

The trigonometric polynomial
\begin{equation}\label{eq:D11}
 p(\varphi)=c_0+2\sum_{m=1}^tc_m\cos(m\varphi)
\end{equation}
is nonnegative. Indeed, if $K_m(\varphi)=(m+1)^{-1}|\sum_{k=0}^me^{ik\varphi}|^2$ is the $m$th Fej\'er kernel, telescoping its Fourier coefficients gives
\begin{equation}\label{eq:D12}
 p(\varphi)=\sum_{m=0}^t(m+1)\Delta^2c_mK_m(\varphi)\geq0.
\end{equation}
The function $w$ is analytic on a neighborhood of the closed unit disk, and $|w(e^{i\varphi})|=1$. Coefficient extraction in Eq.~\eqref{eq:D9} therefore gives
\begin{equation}\label{eq:D13}
 \lambda_j=H^{-L}\int_0^{2\pi}p(\varphi)w(e^{i\varphi})^j\frac{d\varphi}{2\pi}.
\end{equation}
Consequently $|\lambda_j|\leq H^{-L}c_0=\lambda_0$. The vacuum sector attains this value, and $H^{-L}c_0=p_{L,t}(a)$, proving Eq.~\eqref{eq:D7}.

For $t=1$, put $A=1+L\theta$. Direct extraction gives $H^L\lambda_0=A$ and, for $j\geq1$,
\begin{equation}\label{eq:D14}
 H^L\lambda_j=(-\theta)^{j-1}[j(1-\theta^2)-A\theta].
\end{equation}
The finite geometric sum implies $j(1-\theta)\theta^{j-1}\leq1-\theta^j$. Since $A\geq1+\theta$, the triangle inequality yields
\begin{align}
 H^L|\lambda_j|&\leq(1+\theta)(1-\theta^j)+A\theta^j\label{eq:D15}\\
 &\leq A.\label{eq:D16}
\end{align}
The vacuum again attains the bound. This proves the unconditional formula at $t=1$.
\end{proof}

For every integer $t\geq0$ and $0<r<1$, Cauchy's estimate gives
\begin{equation}\label{eq:D17}
 \|Q_{L,t}^\Gamma\|_\infty\leq\frac{r^{-t}}{(1-r)[H(1-\theta r)]^L}.
\end{equation}
Indeed, Eq.~\eqref{eq:C9} bounds the generating operator on $|z|=r$, and Cauchy's coefficient formula applied after multiplication by $(1-z)^{-1}$ gives this inequality. If $0<\alpha=t/L<a$, its exponential factor is minimized at $r=\alpha/[\theta(1+\alpha)]$, giving
\begin{equation}\label{eq:D18}
 \|Q_{L,t}^\Gamma\|_\infty\leq\frac{2^{-LF_H(\alpha)}}{1-r}.
\end{equation}

\begin{appprop}[Achievable state rate]\label{prop:D3}
For $M\geq0$ and $0\leq u<H-1$,
\begin{equation}\label{eq:D19}
 D_{\mathrm{PPT}}(\rho_{H,M,u})\geq F_H(u).
\end{equation}
For finite $L,t$, the threshold fidelity is
\begin{equation}\label{eq:D20}
 p_{L,t}(u)=\sum_{s=0}^t\binom{L+s-1}{s}\frac{u^s}{(1+u)^{L+s}}.
\end{equation}
\end{appprop}
\begin{proof}
The error modes have independent geometric photon distributions. Counting the arrangements of $s$ photons in $L$ modes gives Eq.~\eqref{eq:D20}, independently of $M$.

For $u>0$, choose $u<\alpha<a$ and put $t_L=\lfloor\alpha L\rfloor$. For sufficiently large $L$, $t_L\leq a(L-1)$, so choose
\begin{equation}\label{eq:D21}
 K_L=\left\lfloor p_{L,t_L}(a)^{-1}\right\rfloor.
\end{equation}
Proposition~\ref{prop:D2} and Lemma~\ref{lem:D1} make the map PPT whenever $K_L\geq2$. The summands in $p_{L,t_L}(a)$ increase up to $t_L$. Thus
\begin{equation}\label{eq:D22}
 \Pr(Z_{L,a}=t_L)\leq p_{L,t_L}(a)\leq(t_L+1)\Pr(Z_{L,a}=t_L).
\end{equation}
Stirling's formula now gives $-L^{-1}\log_2p_{L,t_L}(a)\to F_H(\alpha)$. The function $F_H$ vanishes at $a$ and has derivative $\log_2[x/(\theta(1+x))]<0$ for $0<x<a$, so this limit is positive and $K_L\geq2$ eventually.

The total error count has mean $Lu$ and variance $Lu(1+u)$. Chebyshev's inequality gives $p_{L,t_L}(u)\to1$, while $L^{-1}\log_2K_L\to F_H(\alpha)$. Let $\alpha\downarrow u$ to prove the claim. For $u=0$, take $t=0$ and $K_L=\lfloor H^L\rfloor$. The fidelity is one and the rate tends to $\log_2H$.
\end{proof}
The exponent has a direct counting interpretation. For $H=G$ and $\theta=\eta$,
\begin{equation}\label{eq:D23}
 LF_G(t/L)=L\log_2G+t\log_2(1/\eta)-Lg(t/L).
\end{equation}
For $t/L\to\alpha>0$, Stirling's formula gives $L^{-1}\log_2\binom{L+t}{t}\to g(\alpha)$. The entropy term counts the photon arrangements. This interpretation does not assume that the individual error labels are measured or announced.

\subsection{Optimal tests depending on total error number}
For $f:\mathbb N_0\to[0,1]$, consider
\begin{equation}\label{eq:D24}
 Q_f=S_H^{\otimes L}(I_{c^L}\otimes f(N_d))(S_H^\dagger)^{\otimes L}.
\end{equation}
Its fidelity on $\rho_{H,M,u}^{\otimes L}$ is $\mathbb E f(Z_{L,u})$. Since $(S_H^\dagger)^{\otimes L}$ sends the product vacuum to a state with error-mode marginal $\tau_a^{\otimes L}$,
\begin{equation}\label{eq:D25}
 \langle0,0|Q_f|0,0\rangle=\mathbb E f(Z_{L,a}).
\end{equation}
Partial transpose preserves this vacuum expectation, so PPT feasibility requires
\begin{equation}\label{eq:D26}
 \mathbb E f(Z_{L,a})\leq1/K.
\end{equation}

\begin{appprop}[Optimality within the number-test family]\label{prop:D4}
Let $0<u<a$ and $K\geq2$. Choose $t\geq0$ so that $p_{L,t-1}(a)\leq1/K\leq p_{L,t}(a)$, and put
\begin{equation}\label{eq:D27}
 q=\frac{K^{-1}-p_{L,t-1}(a)}{\Pr(Z_{L,a}=t)}.
\end{equation}
The function
\begin{equation}\label{eq:D28}
 f_*(k)=\mathbf1_{\{k<t\}}+q\mathbf1_{\{k=t\}}
\end{equation}
maximizes fidelity subject to Eq.~\eqref{eq:D26}. If the partial transposes of both neighbouring threshold operators $Q_{L,t-1}$ and $Q_{L,t}$ have operator norm equal to their vacuum eigenvalue, then $Q_{f_*}$ is PPT feasible and maximizes fidelity among all PPT-feasible tests in this family. Here $Q_{L,-1}=0$. In particular, this conclusion holds when $t\leq a(L-1)$, and also when $t\leq1$.
\end{appprop}
\begin{proof}
The ratio of the two probability masses is
\begin{equation}\label{eq:D29}
 \frac{\Pr(Z_{L,u}=k)}{\Pr(Z_{L,a}=k)}=
 \left(\frac{1+a}{1+u}\right)^L\left(\frac{u(1+a)}{a(1+u)}\right)^k.
\end{equation}
This ratio $\ell_k$ decreases strictly with $k$, so $(f_*(k)-f(k))(\ell_k-\ell_t)\geq0$ at every $k$. Multiply by $\Pr(Z_{L,a}=k)$ and sum. Since $f_*$ saturates the budget, $\mathbb E f_*(Z_{L,u})\geq\mathbb E f(Z_{L,u})$.

The operator $Q_{f_*}$ equals $(1-q)Q_{L,t-1}+qQ_{L,t}$. Under the stated norm hypothesis, convexity gives
\begin{align}
 \|Q_{f_*}^\Gamma\|_\infty&\leq(1-q)p_{L,t-1}(a)+qp_{L,t}(a)\label{eq:D30}\\
 &=1/K.\label{eq:D31}
\end{align}
Its vacuum expectation is $1/K$, so equality holds and the classical optimum is PPT feasible.
\end{proof}
For $u=0$, only $f(0)$ affects the fidelity. Its optimal value is $\min\{1,H^L/K\}$, attained by $f(k)=\min\{1,H^L/K\}\mathbf1_{\{k=0\}}$. This test is PPT feasible by the exact norm at $t=0$.

\subsection{Finite blocks followed by hashing}
The output in Eq.~\eqref{eq:D2} is isotropic, meaning invariant under $U\otimes\bar U$. Its marginals are maximally mixed and its eigenvalues are $p$ and $(1-p)/(K^2-1)$, the latter repeated $K^2-1$ times. The one-way hashing inequality~\cite{r8} gives the rate per original copy
\begin{equation}\label{eq:D32}
 r_{L,t,K}=\frac{[\log_2K-h_2(p)-(1-p)\log_2(K^2-1)]_+}{L},
\end{equation}
where $p=p_{L,t}(u)$. Hashing acts on many independent outputs of the finite block.

For $\eta=0.8$, $n=0.5$, $H=G=5$, $L=512$ and $t=305$, the condition $t\leq(H-1)(L-1)$ holds. The exact norm gives
\begin{align}
 p&\simeq0.9927,\label{eq:D33}\\
 -\log_2\|Q_{512,305}^\Gamma\|_\infty&\simeq513.15.\label{eq:D34}
\end{align}
Thus $K=2^{513}$ is feasible and yields
\begin{equation}\label{eq:D35}
 r_{512,305,2^{513}}\simeq0.9872>I_{\mathrm{RC}}(0.8,0.5)\simeq0.9445.
\end{equation}
Here $I_{\mathrm{RC}}$ is the reverse-coherent-information benchmark in Eq.~\eqref{eq:O3}. The Cauchy bound gives the lower rate $\simeq0.9757$. Choosing the general integer $K=\lfloor p_{512,305}(4)^{-1}\rfloor$ instead gives $\simeq0.9875$. The limiting capacity is $\simeq1.1054$. The displayed rates use the high-energy limiting parameters. Since the PPT inequality has strict slack for $K=2^{513}$, Lemma~\ref{lem:B1} and continuity imply that sufficiently large finite input energy still gives a strict improvement over the benchmark.

\section{Uniform Rains bounds}\label{app:E}
This appendix matches the achievable exponent by bounds that are uniform in the number of copies. The lower bound is a feasible PPT witness; the upper bound is the relative entropy to an explicit separable state. Matching these two bounds uniformly before regularization avoids exchanging an uncontrolled copy-number limit with an energy limit. Partial-transpose norms also underpin the computable entanglement bounds of Refs.~\cite{r41,r42}.

\subsection{A tensor-product witness}
\begin{applemma}[Witness inequality]\label{lem:E1}
Let $W\geq0$ be bounded, with bounded partial transpose satisfying $\|W^\Gamma\|_\infty\leq1$. Then
\begin{equation}\label{eq:E1}
R(\omega)\geq\operatorname{Tr}(\omega\log_2 W),
\end{equation}
whenever the right-hand side is defined, allowing $-\infty$. The same inequality holds on tensor powers with $W^{\otimes\ell}$.
\end{applemma}
\begin{proof}
For every feasible $\sigma$ in the Rains minimization, trace duality gives
\begin{equation}\label{eq:E2}
\operatorname{Tr}(\sigma W)=\operatorname{Tr}(\sigma^\Gamma W^\Gamma)
\leq\|\sigma^\Gamma\|_1\|W^\Gamma\|_\infty\leq1.
\end{equation}
On finite local number cutoffs the pairing is a finite-dimensional identity. As the cutoffs increase strongly to the identity, the compressed $\sigma$ and $\sigma^\Gamma$ converge in trace norm. Boundedness of $W$ and $W^\Gamma$ then passes both pairings to the limit.

For a bounded test $T\geq\epsilon I>0$, measuring its spectral projections and using relative-entropy data processing yields
\begin{equation}\label{eq:E3}
D(\omega\|\sigma)\geq\operatorname{Tr}(\omega\log_2 T)-\log_2\operatorname{Tr}(\sigma T).
\end{equation}
For a finite spectrum this is the classical log-sum inequality: for measured distributions $p,q$ and positive test values $t_j$, one has $D(p\|q)\geq\sum_jp_j\log_2t_j-\log_2\sum_jq_jt_j$. Finite spectral partitions and their limits cover a continuous spectrum.

Apply this with $T=W+\epsilon I$. Since $\operatorname{Tr}\sigma\leq\|\sigma^\Gamma\|_1\leq1$, the last trace is at most $1+\epsilon$. Let $\epsilon\downarrow0$ and take the infimum over $\sigma$. Finally, $\|(W^{\otimes\ell})^\Gamma\|_\infty=\|W^\Gamma\|_\infty^\ell\leq1$, proving the tensor statement.
\end{proof}

\begin{appprop}[Uniform sandwich]\label{prop:E2}
Fix $H>1$ and $0\leq u<H-1$. For $M>H-1$ and $w>0$, define
\begin{align}
w_M&=\frac{(H-1)(M+1)}{M-H+1},\label{eq:E4}\\
d(u\|w)&=\log_2(w+1)+u\log_2\frac{w+1}{w}-g(u).\label{eq:E5}
\end{align}
For every integer $\ell\geq1$,
\begin{equation}\label{eq:E6}
\ell F_H(u)\leq R(\rho_{H,M,u}^{\otimes\ell})
\leq E_R(\rho_{H,M,u}^{\otimes\ell})\leq\ell d(u\|w_M).
\end{equation}
Moreover, $w_M\to H-1$ and $d(u\|w_M)\to F_H(u)$ as $M\to\infty$.
\end{appprop}
\begin{proof}
For $u>0$ set $x=u/[\theta(1+u)]$, $A=H/(1+u)$, and
\begin{equation}\label{eq:E7}
W=AS_H(I\otimes x^{d^\dagger d})S_H^\dagger.
\end{equation}
Here $0<x<1$. The prefactor in Eq.~\eqref{eq:C3} cancels because $A/[H(1-\theta x)]=1$. The vacuum of $v$ attains the norm, so $\|W^\Gamma\|_\infty=1$. Also
\begin{align}
\operatorname{Tr}(\rho_{H,M,u}\log_2W)
&=\log_2\frac{H}{1+u}+u\log_2\frac{u}{\theta(1+u)}\label{eq:E8}\\
&=F_H(u).\label{eq:E9}
\end{align}
Lemma~\ref{lem:E1} gives the lower bound on every tensor power. For $u=0$, use $W=HS_H(I\otimes|0\rangle\langle0|)S_H^\dagger$. Its logarithm equals $\log_2H$ on the support of the resource state, and the same regularization applies.

For the upper bound choose
\begin{equation}\label{eq:E10}
\sigma_M=S_H(\tau_M\otimes\tau_{w_M})S_H^\dagger.
\end{equation}
For the covariance entries in Eq.~\eqref{eq:B12}, replacing $u$ by $w$ gives
\begin{equation}\label{eq:E11}
(a-1)(b-1)-c^2=4\{Mw-(H-1)(M+w+1)\}.
\end{equation}
The choice $w=w_M$ makes this zero. Both diagonal entries $a-1,b-1$ are nonnegative, hence $V_{\sigma_M}-I_4\geq0$. The state is a Gaussian mixture of coherent product states, with displacement covariance $V_{\sigma_M}-I_4$, and is separable. A limit covers the singular covariance.

Unitary invariance and the common first thermal factor give
\begin{equation}\label{eq:E12}
D(\rho_{H,M,u}\|\sigma_M)=D(\tau_u\|\tau_{w_M})=d(u\|w_M).
\end{equation}
The product comparison state gives the upper bound for each $\ell$. Every separable state is feasible for the Rains bound, proving the middle inequality. Finally,
\begin{equation}\label{eq:E13}
d(u\|H-1)=\log_2H+u\log_2\frac{H}{H-1}-g(u)=F_H(u).
\end{equation}
\end{proof}

\subsection{Distillation and regularization}
\begin{applemma}[Output bound]\label{lem:E3}
For a state $\omega$ on $K\times K$, let $f=\operatorname{Tr}(\Phi_K\omega)$. Then
\begin{equation}\label{eq:E14}
R(\omega)\geq f\log_2K-h_2(f).
\end{equation}
Consequently, $D_{\mathrm{PPT}}(\rho)\leq R^\infty(\rho)$.
\end{applemma}
\begin{proof}
For feasible $\sigma$ put $q=\operatorname{Tr}(\Phi_K\sigma)$ and $s=\operatorname{Tr}\sigma\leq1$. Since $\Phi_K^\Gamma=V_K/K$, trace duality gives $q\leq1/K$. Measuring $\{\Phi_K,I-\Phi_K\}$ yields
\begin{align}
D(\omega\|\sigma)&\geq f\log_2(f/q)+(1-f)\log_2\frac{1-f}{s-q}\label{eq:E15}\\
&\geq f\log_2K-h_2(f).\label{eq:E16}
\end{align}
Zero probabilities are handled by limits. Take the infimum over $\sigma$. For an $L$-copy distillation protocol, monotonicity bounds the output Rains quantity by $R(\rho^{\otimes L})$. Divide by $L$ and let fidelity tend to one. Subadditivity makes $R(\rho^{\otimes L})/L$ converge to $R^\infty(\rho)$.
\end{proof}

\begin{appcor}[Uniform energy limits]\label{cor:E4}
For fixed $H>1$ and $0\leq u<H-1$, each of $D_{\mathrm{PPT}},R^\infty,E_R^\infty,R$ and $E_R$, evaluated on $\rho_{H,M,u}$, tends to $F_H(u)$ as $M\to\infty$. The conclusion also holds along sequences $H_N\to H$, $u_N\to u$ and $M_N\to\infty$.
\end{appcor}
\begin{proof}
Combine Proposition~D.3 and Proposition~\ref{prop:E2} with Lemma~\ref{lem:E3}. After dividing Eq.~\eqref{eq:E6} by $\ell$, both endpoints are independent of $\ell$ and converge to $F_H(u)$. Taking the infimum over $\ell$ therefore preserves the bounds. The hypotheses hold eventually along the stated sequences, and $w_{M_N}(H_N)\to H-1$. Continuity of the endpoint expressions completes the proof, including $u=0$ by the zero-noise witness.
\end{proof}

\subsection{Channel-state consequence}
The preceding uniform bounds also identify the high-energy limits of all five channel-state entanglement quantities.
\begin{appcor}[Common channel-state limits]\label{cor:E5}
For each channel in Theorems~1--3, away from the noiseless identity endpoint, the following energy-constrained state limits exist:
\begin{equation}\label{eq:E17}
\begin{aligned}
\lim_{N\to\infty}D_{\mathrm{PPT}}(\chi_N^{\mathcal N})
&=\lim_{N\to\infty}R^\infty(\chi_N^{\mathcal N})
=\lim_{N\to\infty}E_R^\infty(\chi_N^{\mathcal N})\\
&=\lim_{N\to\infty}R(\chi_N^{\mathcal N})
=\lim_{N\to\infty}E_R(\chi_N^{\mathcal N})\\
&=Q_{\mathrm{PPT}}(\mathcal N).
\end{aligned}
\end{equation}
\end{appcor}
\begin{proof}
For thermal loss, apply Corollary~\ref{cor:E4} to the exact parameters in Lemma~B.1; their limits give the rate in Theorem~1. The relabeled diagonalization in Appendix~\ref{app:G} proves the amplifier case. For additive noise, the local-attenuation lower bound and the separable-comparison upper bound in Appendix~\ref{app:H} squeeze all five quantities to the value in Theorem~3. Entanglement-breaking channel states are separable, so every quantity is zero there.
\end{proof}
All states in these limits are normalized physical states. No infinite-energy density operator is asserted. The energy-constrained theorems identify $D_{\mathrm{PPT}},R^\infty$ and $E_R^\infty$ without requiring single-copy equality of $R$ and $E_R$.

\section{Adaptive thermal-loss converse}\label{app:F}
This appendix completes the thermal-loss theorem by controlling teleportation simulation and taking the energy and block limits in a fixed order. The simulator consumes one bipartite resource state per channel use. Once every use is replaced, the entire adaptive protocol is a single free map acting on the resource copies. The comparison-state bound then limits the output entanglement, irrespective of the memories and classical feedback inside that free map.

\subsection{Uniform simulation with finite-resource energy}
For a difference of channels $\Delta$, define its diamond norm by $\|\Delta\|_\diamond=\sup_{R,\omega}\|(\mathrm{id}_R\otimes\Delta)(\omega)\|_1$, allowing arbitrary reference systems and input states. Continuous-variable teleportation with resource $|\psi_N\rangle\langle\psi_N|$ implements additive noise with mean added photon number
\begin{equation}\label{eq:F1}
\delta_N=2N+1-2\sqrt{N(N+1)}\longrightarrow0.
\end{equation}
The protocol measures commuting quadrature differences, communicates the outcomes and displaces Bob's mode, so it is LOCC~\cite{r10}. With resource $\chi_N$ and displacement gain $\sqrt\eta$, it simulates
\begin{equation}\label{eq:F2}
\mathcal L_{\eta,n}^{(N)}=\mathcal L_{\eta,n}\circ\mathcal Z_{\delta_N}
=\mathcal L_{\eta,n+\eta\delta_N/(1-\eta)}.
\end{equation}
This follows by propagating the random displacement through the beam splitter, or by comparing the action on displacement operators~\cite{r11,r43}.

Both channels in Eq.~\eqref{eq:F2} use the same beam-splitter dilation; only the environmental thermal state differs. Contractivity under the dilation and partial trace gives
\begin{equation}\label{eq:F3}
\begin{aligned}
\epsilon_N&:=\|\mathcal L_{\eta,n}^{(N)}-\mathcal L_{\eta,n}\|_\diamond\\
&\leq\|\tau_{n+\eta\delta_N/(1-\eta)}-\tau_n\|_1\longrightarrow0.
\end{aligned}
\end{equation}
The geometric probabilities in Eq.~\eqref{eq:A6} prove convergence, including $n=0$. Here $N$ is the simulator's resource energy. The diamond bound is uniform over input energies and reference systems~\cite{r43}.

Replace the uses of the channel in a fixed $m$-use adaptive protocol by this simulation. All teleportation steps are PPT, so they combine with the original operations into a PPT map $\Lambda_{m,N}$. If $\omega_m$ and $\omega_{m,N}$ denote the true and simulated outputs, then
\begin{equation}\label{eq:F4}
\omega_{m,N}=\Lambda_{m,N}(\chi_N^{\otimes m}),~~
\|\omega_{m,N}-\omega_m\|_1\leq m\epsilon_N.
\end{equation}
The second inequality follows by replacing one use at a time and applying trace-norm contractivity. Each replacement contributes at most $\epsilon_N$, even for the correlated input at that stage. Subsequent channels and adaptive operations cannot increase that contribution, so the $m$ contributions add. Classical branches are already included as registers.

\subsection{Completion of the capacity proof}
\begin{proof}[Proof of Theorem 1]
First suppose $n<G-1$. For sufficiently large $N$, Proposition~\ref{prop:E2} applies to the parameters of Lemma~B.1. Put
\begin{equation}\label{eq:F5}
b_N^{\mathrm{ref}}=d(u_N\|w_{M_N}(H_N))\longrightarrow B_{\mathrm{loss}}(\eta,n).
\end{equation}
By Appendix~\ref{app:A}, it suffices to bound distribution protocols. Fix one, with target dimension $K_m$, output $\omega_m$ and overlap $f_m=\operatorname{Tr}(\Phi_{K_m}\omega_m)$ tending to one. Write $f_{m,N}=\operatorname{Tr}(\Phi_{K_m}\omega_{m,N})$ for the overlap of the simulated output. The output bound, monotonicity and the product comparison give
\begin{equation}\label{eq:F6}
\begin{aligned}
f_{m,N}\log_2K_m-h_2(f_{m,N})&\leq R(\omega_{m,N})\\
&\leq R(\chi_N^{\otimes m})\leq mb_N^{\mathrm{ref}}.
\end{aligned}
\end{equation}
For this fixed protocol and fixed $m$, first let $N\to\infty$. The target dimension $K_m$ is fixed in this step, so the factor $\log_2K_m$ does not obstruct convergence. Equation~\eqref{eq:F4} gives $f_{m,N}\to f_m$, and hence
\begin{equation}\label{eq:F7}
f_m\log_2K_m-h_2(f_m)\leq mB_{\mathrm{loss}}(\eta,n).
\end{equation}
Only now let $m\to\infty$ along a successful code sequence. Since $f_m\to1$, its rate cannot exceed $B_{\mathrm{loss}}$. The argument uses Rains monotonicity; the separable state merely evaluates an upper bound on that quantity.

For achievability, fix finite $N$, prepare $|\psi_N\rangle$ locally and send its second mode through each use. This produces independent copies of $\chi_N$. Proposition~D.3 gives rate at least $F_{H_N}(u_N)$, which tends to $B_{\mathrm{loss}}$. Teleportation then converts the distilled ebits into a code at the same rate (Appendix~\ref{app:A}). For any strict target rate below this limit, choose sufficiently large finite $N$, then a threshold ratio above $u_N$ and finally take the block length to infinity. Corollary~\ref{cor:E4} gives the channel-state limits of Corollary~\ref{cor:E5} for thermal loss.

If $n\geq\eta/(1-\eta)$, the channel is entanglement breaking~\cite{r32}. Its measure-and-prepare implementation is a free PPT operation. The full protocol therefore preserves PPT, and its overlap with $\Phi_K$ is at most $1/K$ by trace duality. No positive rate is possible. Its channel states are separable, so all five entanglement quantities vanish at every finite energy.
\end{proof}

\section{Thermal amplifiers}\label{app:G}
This appendix proves the amplifier theorem by interchanging the two mode labels in the squeezed-thermal analysis.

For $\chi_N^{\mathcal A_{\kappa,n}}$, the covariance entries are
\begin{align}
a_N&=2N+1,\label{eq:G1}\\
b_N&=2\kappa N+2(\kappa-1)(n+1)+1,\label{eq:G2}\\
c_N&=2\sqrt{\kappa N(N+1)}.\label{eq:G3}
\end{align}
Apply Lemma~B.1 with the roles of $a_N$ and $b_N$ exchanged. This is a relabeling of the normal form. Partial transpose on Alice and partial transpose on Bob differ by full transpose, which preserves the operator norm. Hence the PPT norm estimates are unchanged. The diverging thermal occupation is now associated with Bob's mode. Expansion gives
\begin{align}
d_N&=2(\kappa-1)N+2(\kappa+1)n+2\kappa+O(N^{-1}),\label{eq:G4}\\
H_N&\longrightarrow\frac{\kappa}{\kappa-1},\quad
\theta_N\longrightarrow\frac1\kappa,\quad u_N\longrightarrow n,\label{eq:G5}\\
M_N&=(\kappa-1)N+\kappa n+\kappa-1+O(N^{-1}).\label{eq:G6}
\end{align}
The threshold protocol and uniform sandwich remain valid after this relabeling. The relevant conditions depend only on the two thermal occupations and the squeezing gain. In particular, the finite error occupation again tends to the environmental noise $n$, and the threshold condition $u<H-1$ becomes $n<1/(\kappa-1)$, exactly the non-entanglement-breaking condition. Their common limit is the first line of Eq.~\eqref{eq:amplifier-capacity}.

Teleportation with gain $\sqrt\kappa$ and resource $\chi_N^{\mathcal A_{\kappa,n}}$ simulates
\begin{equation}\label{eq:G7}
\mathcal A_{\kappa,n}\circ\mathcal Z_{\delta_N}
=\mathcal A_{\kappa,n+\kappa\delta_N/(\kappa-1)}.
\end{equation}
As for thermal loss, a common dilation bounds the diamond distance by the trace distance between the environmental thermal states. It tends to zero. The proof in Appendix~\ref{app:F} then gives both the adaptive converse and achievability with a fixed mean input energy. For $n\geq1/(\kappa-1)$ the channel is entanglement breaking, so the capacity and all channel-state entanglement quantities vanish. This proves Theorem~2 and the amplifier part of Corollary~\ref{cor:E5}.

\section{Additive Gaussian noise}\label{app:H}
This appendix proves the additive-noise theorem using local attenuation for achievability and a separate Gaussian comparison state for the converse. The direct squeezed-thermal limit is singular here: both normal-mode occupations grow. A small amount of local loss converts the problem to the already controlled thermal-loss family. The converse can then be evaluated directly on the original additive-noise resource.

\subsection{Achievability by local attenuation}
For any $0<\eta<1$, Bob can apply pure loss after $\mathcal Z_\xi$. The resulting channel obeys
\begin{equation}\label{eq:H1}
\mathcal L_{\eta,0}\circ\mathcal Z_\xi=\mathcal L_{\eta,n_\eta},~~
n_\eta=\frac{\eta\xi}{1-\eta}.
\end{equation}
Both sides take $V$ to $\eta V+(1-\eta+2\eta\xi)I_2$, and their displacement characteristic functions agree. If $\xi<1$, then $n_\eta<\eta/(1-\eta)$, so Theorem~1 gives
\begin{equation}\label{eq:H2}
\begin{aligned}
Q_{\mathrm{PPT}}(\mathcal Z_\xi)\geq{}&-\log_2(1-\eta)-g\!\left(\frac{\eta\xi}{1-\eta}\right)\\
&+\frac{\eta\xi}{1-\eta}\log_2(1/\eta).
\end{aligned}
\end{equation}
As $x\to\infty$, $g(x)=\log_2x+1/\ln2+O(x^{-1})$, and $\eta\ln(1/\eta)/(1-\eta)\to1$. Letting $\eta\uparrow1$ in this numerical bound yields $-\log_2\xi+(\xi-1)/\ln2$.

Operationally, first fix $\eta<1$ close enough to one, then choose finite input energy, and only then let the number of channel uses grow.

\subsection{A separable comparison state}
Write $\chi_N^\xi=\chi_N^{\mathcal Z_\xi}$. Its covariance entries are
\begin{equation}\label{eq:H3}
a=2N+1,\quad b=2N+2\xi+1,\quad c=2\sqrt{N(N+1)}.
\end{equation}
Choose a Gaussian state $\sigma_N^\xi$ with the same $a,b$ and correlation
\begin{equation}\label{eq:H4}
c_s=\sqrt{(a-1)(b-1)}=2\sqrt{N(N+\xi)}.
\end{equation}
Its covariance minus $I_4$ is positive. Thus it is a mixture of coherent product states and is separable.

Set $s=\sqrt{4\xi N+(1+\xi)^2}$ and $s_s=\sqrt{4N+(1+\xi)^2}$. The thermal occupations in the diagonal form of the resource are $x_1,x_2$, and those of the comparison are $y_1,y_2$, where
\begin{align}
x_1&=(s-\xi-1)/2,\quad x_2=(s+\xi-1)/2,\label{eq:H5}\\
y_1&=(s_s-\xi-1)/2,\quad y_2=(s_s+\xi-1)/2.\label{eq:H6}
\end{align}
For $N>0$, both $y_j$ are positive. The comparison state's squeezing gain is $H_s=(a+b+2s_s)/(4s_s)$. In the modes that diagonalize that state, the resource's mean occupations $e_1,e_2$ satisfy
\begin{align}
2e_1+1&=H_sa+(H_s-1)b-2\sqrt{H_s(H_s-1)}\,c,\label{eq:H7}\\
2e_2+1&=(H_s-1)a+H_sb-2\sqrt{H_s(H_s-1)}\,c.\label{eq:H8}
\end{align}
Expanding the logarithm of the comparison state's thermal factors gives the exact relative entropy
\begin{equation}\label{eq:H9}
\begin{aligned}
D(\chi_N^\xi\|\sigma_N^\xi)={}&-g(x_1)-g(x_2)+\log_2[(y_1+1)(y_2+1)]\\
&+e_1\log_2(1+1/y_1)+e_2\log_2(1+1/y_2).
\end{aligned}
\end{equation}
The logarithm of each comparison thermal factor is affine in its number operator, so its expectation depends only on $e_1,e_2$. The resource entropy uses its own thermal eigenvalues $x_1,x_2$. Thus the calculation applies even though the two states need not share an eigenbasis.

At fixed $\xi>0$, their asymptotics are
\begin{align}
g(x_1)+g(x_2)&=\log_2(\xi N)+2/\ln2+o(1),\label{eq:H10}\\
\log_2[(y_1+1)(y_2+1)]&=\log_2N+o(1),\label{eq:H11}\\
y_j&=\sqrt N+O(1),\label{eq:H12}\\
e_j&=\frac{1+\xi}{2}\sqrt N+O(1).\label{eq:H13}
\end{align}
For the last line, one may use $e_1-e_2=-\xi$ and
\begin{equation}\label{eq:H14}
e_1+e_2=\frac{(a+b)^2-4c_sc}{4s_s}-1.
\end{equation}
Substitution in Eq.~\eqref{eq:H9} gives
\begin{equation}\label{eq:H15}
D(\chi_N^\xi\|\sigma_N^\xi)\longrightarrow-\log_2\xi+\frac{\xi-1}{\ln2}.
\end{equation}

\subsection{Converse and state limits}
Teleportation with $\chi_N^\xi$ simulates $\mathcal Z_{\xi+\delta_N}$. For $\xi>0$, its diamond distance from $\mathcal Z_\xi$ is bounded by the $L^1$ distance of their Gaussian displacement densities in Eq.~\eqref{eq:A9}. This distance tends to zero, by convergence on a compact disk and a bound on the Gaussian tails. Hence the stretching estimate of Eq.~\eqref{eq:F4} applies.

The product comparison state gives, for every $m$,
\begin{equation}\label{eq:H16}
R((\chi_N^\xi)^{\otimes m})\leq mD(\chi_N^\xi\|\sigma_N^\xi).
\end{equation}
The output-fidelity argument from Appendix~\ref{app:F}, followed by Eq.~\eqref{eq:H15}, proves the converse for $0<\xi<1$.

For the state limits, apply the local attenuation of Eq.~\eqref{eq:H1} to Bob's half of $\chi_N^\xi$. It produces the thermal-loss channel state at the same $N$. Therefore, for fixed $\eta$,
\begin{equation}\label{eq:H17}
\liminf_{N\to\infty}D_{\mathrm{PPT}}(\chi_N^\xi)\geq B_{\mathrm{loss}}(\eta,n_\eta).
\end{equation}
Let $\eta\uparrow1$ and use Eq.~\eqref{eq:H15} as the common upper bound. Together with $D_{\mathrm{PPT}}\leq R^\infty\leq E_R^\infty$ and the single-copy bounds, this proves every limit in Corollary~\ref{cor:E5} for additive noise.

For $\xi\geq1$ the channel is entanglement breaking~\cite{r32}, so capacity and channel-state entanglement vanish. At $\xi=0$, arbitrary finite-dimensional subspaces can be transmitted noiselessly without an energy constraint. This proves Theorem~3.

\section{The one-photon branch}\label{app:I}
This appendix gives the partial-transpose spectrum and distillation rate for a supplied resource containing exactly one error photon. Rates are per supplied branch copy.

For $G>1$, let $\eta=(G-1)/G$. Define the branch projection and energy-constrained branch state by
\begin{equation}\label{eq:I1}
\begin{aligned}
P_j&=S_G(I\otimes|j\rangle\langle j|)S_G^\dagger,\\
\rho_M^{(j)}&=S_G(\tau_M\otimes|j\rangle\langle j|)S_G^\dagger.
\end{aligned}
\end{equation}
These branches have orthogonal supports, and
\begin{equation}\label{eq:I2}
\rho_{G,M,n}=\sum_{j=0}^\infty\frac{n^j}{(1+n)^{j+1}}\rho_M^{(j)}.
\end{equation}
The projection $P_1$ selects exactly one photon. It differs from the threshold $Q_{1,1}=P_0+P_1$.

\begin{appprop}[One-photon norm and rate]\label{prop:I1}
The eigenvalues of $P_1^\Gamma$, indexed by the occupation $j$ of mode $v$, are
\begin{align}
\lambda_0&=\eta(1-\eta),\label{eq:I3}\\
\lambda_j&=(1-\eta)(-\eta)^{j-1}[j(1-\eta^2)-\eta^2],\quad j\geq1.\label{eq:I4}
\end{align}
Each is repeated over the occupation of the orthogonal mode. For $\eta\geq1/2$,
\begin{equation}\label{eq:I5}
\|P_1^\Gamma\|_\infty=\eta(1-\eta)\max\{1,2-3\eta^2\}.
\end{equation}
For $\eta\geq1/\sqrt3$,
\begin{equation}\label{eq:I6}
\lim_{M\to\infty}D_{\mathrm{PPT}}(\rho_M^{(1)})=\log_2G+\log_2(1/\eta).
\end{equation}
\end{appprop}
\begin{proof}
Differentiate Eq.~\eqref{eq:C3} at $z=0$. In the $j$th occupation sector the scalar function is $(1-\eta)(z-\eta)^j/(1-\eta z)^{j+1}$. Its derivative gives the stated eigenvalues.

Divide their absolute values by $\eta(1-\eta)$ and put $v_0=\eta^2$. The $j=0$ value is one; the $j=1$ value is $|1-2v_0|/\sqrt{v_0}\leq1$ on $1/4\leq v_0\leq1$; the $j=2$ value is $|2-3v_0|$. For $j\geq3$ the value is
\begin{equation}\label{eq:I7}
v_0^{(j-2)/2}|j-(j+1)v_0|.
\end{equation}
Where the bracket is negative, this increases to one at $v_0=1$. Where it is positive, its unrestricted maximum occurs at $v_0=(j-2)/(j+1)$ and equals $2[(j-2)/(j+1)]^{(j-2)/2}\leq1$. The last inequality is equivalent, for $k=j-2\geq1$, to $(1+3/k)^k\geq4$; the left-hand side increases from its value four at $k=1$. This proves Eq.~\eqref{eq:I5}.

On $(\rho_M^{(1)})^{\otimes L}$, the test $P_1^{\otimes L}$ has fidelity one. For $\eta\geq1/\sqrt3$, its partial-transpose norm is $[\eta(1-\eta)]^L$. Choosing $K_L=\lfloor[\eta(1-\eta)]^{-L}\rfloor$ proves the lower bound in Eq.~\eqref{eq:I6}, for every $M$.

For the upper bound use the separable comparison of Eq.~\eqref{eq:E10}, with $H=G$. Its relative entropy from $\rho_M^{(1)}$ is
\begin{equation}\label{eq:I8}
\log_2(w_M+1)+\log_2\frac{w_M+1}{w_M}
\longrightarrow\log_2G+\log_2(1/\eta).
\end{equation}
The product comparison bounds the regularized Rains quantity and hence $D_{\mathrm{PPT}}$, completing the proof.
\end{proof}
At $\eta=0.8$ the rate in Eq.~\eqref{eq:I6} is $\simeq2.644$ ebits per branch copy. At $\eta=0.5$, Eq.~\eqref{eq:I5} gives the rigorous lower bound $-\log_2(5/16)\simeq1.678$, whereas the comparison-state upper bound is two.

\section{Nonseparability and approximation bounds}\label{app:J}
We prove that every finite error-photon threshold map with target dimension $K\geq2$ is nonseparable, and obtain a positive energy-constrained distance from separable channels for each fixed map. The proof uses finite local probes and Schmidt-number witnesses~\cite{r45,r46}. This obstruction concerns each fixed map. It neither excludes another LOCC construction nor proves a gap between the optimal capacities.

\subsection{A filtered Schmidt-number condition}
The Schmidt rank of a pure bipartite vector is the number of nonzero coefficients in its Schmidt decomposition, equivalently the smallest number of product vectors whose sum is that vector. Let $\mathcal C_K$ be the trace-norm closed cone generated by projectors onto vectors of Schmidt rank at most $K$. A nonzero positive operator has Schmidt number at most $K$ when it belongs to $\mathcal C_K$~\cite{r45,r46}.

We call a channel separable when its extension by arbitrary local ancillary identities preserves separable states. On finite input restrictions this is equivalent to a separable Choi operator, or a product-Kraus decomposition. We include limits for which all such restrictions remain separable. This class contains LOCC and its operational closure.

Let $R_A,R_B$ be positive diagonal finite-rank operators in local number bases, normalized by $\operatorname{Tr}R_A^2=\operatorname{Tr}R_B^2=1$. With $D=R_A\otimes R_B$ and full input transpose $T$, define the finite reference operator
\begin{equation}\label{eq:J1}
F=DQ^T D.
\end{equation}
A Hermitian $Z$ is $K$-block-positive if $\langle\psi|Z|\psi\rangle\geq0$ for every vector of Schmidt rank at most $K$. Define the normalized witness functional
\begin{equation}\label{eq:J2}
\mathcal W_K(F)=\min_{\substack{Z=Z^\dagger,\ \|Z\|_\infty\leq1\\Z\ K\text{-block-positive}}}\operatorname{Tr}(ZF).
\end{equation}
The minimum exists in finite dimension. Since $Z=0$ is allowed, $\mathcal W_K(F)\leq0$. Separation of a point from a closed convex cone gives
\begin{equation}\label{eq:J3}
F\in\mathcal C_K\quad\Longleftrightarrow\quad\mathcal W_K(F)\geq0.
\end{equation}
Indeed, each generator of the cone has nonnegative expectation against every allowed witness. Conversely, normalizing a separating Hermitian functional gives a negative value in Eq.~\eqref{eq:J2}.

\begin{appprop}[Necessary condition]\label{prop:J1}
If $\mathcal R_{Q,K}$ is separable, every filter in Eq.~\eqref{eq:J1} satisfies $F\in\mathcal C_K$ and $\mathcal W_K(F)\geq0$.
\end{appprop}
\begin{proof}
Prepare two local entangled probes
\begin{align}
|\Omega_A\rangle&=\sum_i(R_A)_{ii}|i\rangle_{A_0}|i\rangle_A,\label{eq:J4}\\
|\Omega_B\rangle&=\sum_j(R_B)_{jj}|j\rangle_B|j\rangle_{B_0}.\label{eq:J5}
\end{align}
Their product is separable across the two laboratories. Applying the map to $AB$, with references placed first, gives
\begin{equation}\label{eq:J6}
J_Q=F\otimes\Phi_K+(R_A^2\otimes R_B^2-F)\otimes\frac{I-\Phi_K}{K^2-1}.
\end{equation}
For a separable map, $J_Q$ is separable across $A_0A':B'B_0$. Contract the output systems with $\langle\varphi_K|$ and $|\varphi_K\rangle$. The result is $F$. On a pure product term, write the local vectors as $\sum_{r=1}^K|a_r\rangle_{A_0}|r\rangle_{A'}$ and $\sum_{s=1}^K|s\rangle_{B'}|b_s\rangle_{B_0}$. The contraction is $K^{-1/2}\sum_{r=1}^K|a_r\rangle|b_r\rangle$, of Schmidt rank at most $K$. Mixtures and limits put $F$ in $\mathcal C_K$, and Eq.~\eqref{eq:J3} gives the witness condition. This contraction is a proof device, not a proposed free operation.
\end{proof}

\subsection{The range of a hard threshold}
\begin{applemma}[No finite Schmidt rank]\label{lem:J2}
For $L\geq1$, a finite integer $t\geq0$ and $0<\lambda<1$, let $\mathcal V_{\lambda,t}$ contain the normalizable Bargmann functions
\begin{equation}\label{eq:J7}
f(x,y)=e^{\lambda x\cdot y}\sum_{|j|\leq t}y^j f_j(x),
\end{equation}
where $x,y\in\mathbb C^L$ and all coefficients are entire. This is a closed subspace, and none of its nonzero vectors has finite Schmidt rank across $x:y$. For $\lambda=\sqrt\theta$, $\operatorname{Ran}Q_{L,t}\subseteq\mathcal V_{\lambda,t}$.
\end{applemma}
\begin{proof}
Multiplication by $e^{-\lambda x\cdot y}$ turns Eq.~\eqref{eq:J7} into a polynomial in $y$ of total degree at most $t$. This is equivalent to the vanishing of all derivatives $\partial_y^\beta$ with $|\beta|=t+1$. Bargmann norm convergence implies uniform convergence on compact sets, also for derivatives. These equations persist under norm limits, proving closedness.

Suppose a nonzero function in this subspace had finite Schmidt rank $r$. Then $f(x,y)=\sum_{i=1}^r u_i(x)v_i(y)$ for entire $u_i,v_i$. Fix $x_2,\ldots,x_L$ and $y_2,\ldots,y_L$ so that the polynomial in Eq.~\eqref{eq:J7} is not identically zero in $x_1,y_1$. Such a choice exists. At least one polynomial coefficient is a nonzero entire function of $x_1$, whose zeros are discrete. For infinitely many distinct $x_1=z_j$, the slice in $y_1$ has the form
\begin{equation}\label{eq:J8}
e^{\lambda z_jy_1}p_j(y_1),~~ p_j\ne0,\quad\deg p_j\leq t,
\end{equation}
up to a common nonzero factor. All slices would lie in the span of $r$ fixed slices of the $v_i$.

But any finite collection of Eq.~\eqref{eq:J8} is linearly independent. To isolate its $j$th member in a putative relation, apply $\prod_{i\ne j}(\partial_{y_1}-\lambda z_i)^{t+1}$. It kills every other member. On the remaining polynomial, each factor has a nonzero constant term after the exponential is removed, and is invertible on polynomials of degree at most $t$. Thus the $j$th coefficient must vanish, contradicting finite rank.

Finally, apply the kernel Eq.~\eqref{eq:C4} to an arbitrary useful-mode function and an error-mode monomial of total degree at most $t$. The result is $e^{\lambda x\cdot y}$ times a polynomial in $y$ of that degree, with entire coefficients in $x$. Finite sums and norm limits prove the range inclusion.
\end{proof}
A state in $\mathcal C_K$ admits a probability-measure decomposition into pure states of Schmidt rank at most $K$~\cite[Proposition~3A]{r46}. This decomposition need not be countable. If the state is supported in a closed subspace $\mathcal V$, almost every vector in the measure belongs to $\mathcal V$: integrate the nonnegative quantity $\langle\psi|\Pi_{\mathcal V^\perp}|\psi\rangle$. It follows that no nonzero state supported in $\mathcal V_{\lambda,t}$ belongs to any $\mathcal C_K$.

For completeness, this measure statement can also be obtained from the closed-cone definition. Approximate a state in $\mathcal C_K$ by finite ensembles of rank-at-most-$K$ vectors. Trace-norm convergence of the barycenters makes the measures tight. To see this, choose increasing finite-rank projections with uniformly small expected tails. Markov's inequality confines all but arbitrarily small measure to vectors with uniformly vanishing tails; this set is norm compact in projective Hilbert space. A weakly convergent subsequence has the required barycenter. Rank-at-most-$K$ vectors form a closed set, so the limiting measure has the required support. Restricting the argument to countable decompositions would not suffice.

\begin{appthm}[Finite threshold maps are nonseparable]\label{thm:J3}
For every $H>1$, $L\geq1$, finite $t\geq0$ and integer $K\geq2$, the map $\mathcal R_{Q_{L,t},K}$ is nonseparable. Some finite-rank local filters satisfy $\mathcal W_K(F)<0$. This holds whether or not the map obeys the PPT criterion.
\end{appthm}
\begin{proof}
First use filters without a photon cutoff. Fix $0<q<1$ and let $N_A,N_B$ be the total local photon numbers. Set
\begin{equation}\label{eq:J9}
R_A=(1-q)^{L/2}q^{N_A/2},~~ R_B=(1-q)^{L/2}q^{N_B/2}.
\end{equation}
Both squared filters have trace one. The real squeezing convention gives $Q_{L,t}^T=Q_{L,t}$. Thus $F=DQ_{L,t}D$ is nonzero, positive and trace class. Its support is the closure of $\operatorname{Ran}(DQ_{L,t})$, since $F=(DQ_{L,t})(DQ_{L,t})^\dagger$.

In the Bargmann representation, $D$ replaces $x,y$ by $\sqrt q\,x,\sqrt q\,y$, up to a scalar. Lemma~\ref{lem:J2} puts its range in the closed subspace $\mathcal V_{q\sqrt\theta,t}$. That subspace has no nonzero finite-Schmidt-rank vectors. The probability-measure argument therefore places $F$ outside $\mathcal C_K$ for every finite $K$.

Truncate each filter to at most $d$ photons per local block and renormalize the trace of its square to one. The resulting finite operator $F_d$ converges to $F$ in trace norm, by local compression and convergence of the normalizations. If all $F_d$ belonged to $\mathcal C_K$, closedness would put $F$ there too. Thus some finite $d$ has $F_d\notin\mathcal C_K$, and Eq.~\eqref{eq:J3} gives $\mathcal W_K(F_d)<0$. Proposition~\ref{prop:J1} rules out separability of the map.
\end{proof}

\subsection{A quantitative energy-constrained obstruction}
\begin{appprop}[Distance from separable channels]\label{prop:J4}
For a finite filtered operator~\eqref{eq:J1}, suppose $\mathcal W_K(F)=-\delta<0$. Every separable channel $\Sigma$ with the same input and output spaces then satisfies
\begin{equation}\label{eq:J10}
\frac12\|\mathcal R_{Q,K}-\Sigma\|_\diamond\geq\delta.
\end{equation}
The same bound holds with the diamond supremum restricted to an energy budget that includes the finite-support probes in Eqs.~\eqref{eq:J4} and~\eqref{eq:J5}. If a standard-form map $\mathcal R_{T,K}$ is separable, then also
\begin{equation}\label{eq:J11}
\|Q-T\|_\infty\geq\delta.
\end{equation}
\end{appprop}
\begin{proof}
Choose a minimizing witness $Z$ in Eq.~\eqref{eq:J2}. Apply the probes to both maps and denote the output of $\Sigma$ by $J_\Sigma$. Let $F_\Sigma=\langle\varphi_K|J_\Sigma|\varphi_K\rangle$ act on the references. The contraction argument gives $F_\Sigma\in\mathcal C_K$, hence
\begin{equation}\label{eq:J12}
\delta\leq\operatorname{Tr}[Z(F_\Sigma-F)]\leq\|F_\Sigma-F\|_1.
\end{equation}
Twirl both output systems by the LOCC channel averaging $U\otimes\overline U$. This leaves $J_Q$ unchanged. Both maps preserve trace, so their reference marginals are $R_A^2\otimes R_B^2$. The difference of their twirled outputs is therefore
\begin{equation}\label{eq:J13}
(F-F_\Sigma)\otimes\left(\Phi_K-\frac{I-\Phi_K}{K^2-1}\right).
\end{equation}
The two output states have orthogonal supports and trace one. Its trace norm is $2\|F-F_\Sigma\|_1$. Contractivity under twirling proves Eq.~\eqref{eq:J10}, including its restricted-energy version.

For $\Sigma=\mathcal R_{T,K}$, the reference coefficient is $DT^TD$. The Schatten norm inequality gives
\begin{equation}\label{eq:J14}
\|D(Q-T)^TD\|_1\leq\|Q-T\|_\infty\operatorname{Tr}D^2=\|Q-T\|_\infty,
\end{equation}
which proves the last claim.
\end{proof}
The witness gives a strictly positive bound for every fixed finite instance. It is not a closed formula for the optimal distance as a function of block size, threshold and output dimension, and need not stay positive along an asymptotic code sequence. It also does not rule out approximation only on the resource states used by a protocol.

The obstruction does not depend on the thermal noise. It holds at $n=0$, including threshold $t=0$. Yet $Q_2(\mathcal L_{\eta,0})=-\log_2(1-\eta)$~\cite{r11}. Nonseparability of this construction therefore does not prove that LOCC has a smaller noisy-channel capacity.

\section{Photon sectors and exact Rains witnesses}\label{app:K}

This appendix supplies the state bounds for the energy-constrained extension. We use the channel $\mathcal G_{\tau,v}$ and energy function $F_{\tau,v}$ defined in Sec.~III. For the strictly noisy case, assume
\begin{equation}\label{eq:K1}
 \max\{0,\tau-1\}<v<\tau,~~
 z=\frac{\tau}{v(1+v-\tau)}>1.
\end{equation}
Thermal loss corresponds to $\tau=\eta$ and $v=(1-\eta)n$.

The purpose of the sector resolution is to turn an infinite-dimensional Gaussian-state problem into finite-dimensional blocks without losing the asymptotic rate. Local photon counts cost only a sublinear amount of entropy. Within each block, passive-mode symmetry makes the resource and its optimal comparison simultaneously diagonal, allowing an exact Rains witness.

\subsection{Photon-sector spectra}

For $s$ modes, let $\mathcal H_{s,m}\simeq\operatorname{Sym}^{m}\mathbb C^s$ be the subspace with total photon number $m$. Its dimension is
\begin{equation}\label{eq:K2}
 d_{s,m}=\binom{s+m-1}{m},~~ d_{s,-1}=0.
\end{equation}
The binomial coefficient counts ways to distribute $m$ indistinguishable photons among $s$ modes. Equivalently, the mode labels form the symmetric $m$-particle representation of $\mathbb C^s$. Prepare a maximally entangled state on two copies of $\mathcal H_{s,m}$ and send Bob's share through $s$ channel uses. Denote the resulting state by $J_{s,m}$. Bob's output photon count $k$ is a local classical flag; write $\rho_{s,m,k}$ for the conditional state.

Fix $s\geq2$, $m,k\geq0$, and put $q=\min(m,k)$ and $A=s+m+k-1$. For reference and output annihilation operators $a_i,b_i$, define
\begin{equation}\label{eq:K3}
 C_-=\sum_{i=1}^{s}a_ib_i,~~ C_+=C_-^\dagger,~~
 C_j=C_+^jC_-^j.
\end{equation}
The commutator $[C_-,C_+]=s+\widehat N_A+\widehat N_B$ gives orthogonal ladder spaces with projectors $\Pi_h$, $0\leq h\leq q$. Each is the image under $C_+^h$ of the kernel of $C_-$ at photon numbers $(m-h,k-h)$. For a nonnegative integer $a$, write $a^{\underline j}=a(a-1)\cdots(a-j+1)$ and $a^{\underline0}=1$. The dimensions and eigenvalues are
\begin{align}
 D_h&=d_{s,m-h}d_{s,k-h}-d_{s,m-h-1}d_{s,k-h-1},\label{eq:K4}\\
 C_j\Pi_h&=h^{\underline j}(A-h)^{\underline j}\Pi_h.\label{eq:K5}
\end{align}
Indeed, the commutator makes the raising map injective. Subtracting the dimension of its preceding number space gives Eq.~\eqref{eq:K4}; iterating the commutator gives Eq.~\eqref{eq:K5}. The dimensions sum to $d_{s,m}d_{s,k}$. The ladder label $h$ counts collective pair raising from a vector annihilated by $C_-$. It quantifies a correlation sector, rather than a local photon count or an environmental loss event. A projector onto a range of $h$ will define the distillation test, but its implementation is justified through the PPT norm condition.

Define the polynomial $H_h(z)$ and normalizing constant $Z_{s,m,k}$ by
\begin{equation}\label{eq:K6}
 \begin{split}
 H_h(z)&=\sum_{j=0}^{h}\binom hj\binom{A-h}{j}z^j,\\
 Z_{s,m,k}&=\sum_{h=0}^{q}D_hH_h(z).
 \end{split}
\end{equation}
The conditional state and a separable comparison have the spectral forms
\begin{equation}\label{eq:K7}
 \rho_{s,m,k}=\frac{1}{Z_{s,m,k}}\sum_{h=0}^{q}H_h(z)\Pi_h,
 ~~ \sigma_{s,m,k}=\sum_{h=0}^{q}\mu_h\Pi_h,
\end{equation}
where
\begin{equation}\label{eq:K8}
 \mu_h=\frac{m!k!(s-1)!}{h!(A-h)!}>0.
\end{equation}
To describe the comparison directly, let $e$ be a unit vector in $\mathbb C^s$, let $de$ be normalized invariant measure, and let $|e;r\rangle$ be the condensate of $r$ photons in direction $e$. Then
\begin{equation}\label{eq:K9}
 \begin{split}
 |e;r\rangle&=\frac{(\sum_i e_i a_i^\dagger)^r}{\sqrt{r!}}|0\rangle,\\
 \sigma_{s,m,k}&=\int |\bar e;m\rangle\langle\bar e;m|
 \otimes |e;k\rangle\langle e;k|\,de.
 \end{split}
\end{equation}
The spectra in Eq.~\eqref{eq:K7} are the photon-sector identities of Ref.~\cite{r16}, Appendix~C. The comparison $\sigma_{s,m,k}$ is visibly separable in its condensate representation: it shares a random mode direction but contains only product states. Its spectral form is useful because it commutes with the resource, so the relative entropy is an ordinary weighted sum over the ladder spaces.

\begin{applemma}[Partial swaps]\label{lem:K1}
For every $0\leq j\leq q$,
\begin{equation}\label{eq:K10}
 \|C_j^\Gamma\|_\infty\leq m^{\underline j}k^{\underline j}.
\end{equation}
All these bounds are saturated by the same aligned condensate after partial transpose.
\end{applemma}
\begin{proof}
Let $V$ embed the two symmetric photon spaces into $(\mathbb C^s)^{\otimes m}\otimes(\mathbb C^s)^{\otimes k}$. Let $F_j$ swap the first $j$ tensor factors between the two parties. Expansion of the ordered annihilation strings, followed by transposition of Bob's indices, gives
\begin{equation}\label{eq:K11}
 C_j^\Gamma=m^{\underline j}k^{\underline j}V^\dagger F_jV.
\end{equation}
The falling factorials count ordered particle selections. The operator $F_j$ is a self-adjoint unitary, so its compression has norm at most one. It fixes $e^{\otimes m}\otimes e^{\otimes k}$, which proves simultaneous saturation.
\end{proof}

\begin{appthm}[Exact sector Rains value]\label{thm:K2}
In the region \eqref{eq:K1},
\begin{equation}\label{eq:K12}
 R(\rho_{s,m,k})=E_R(\rho_{s,m,k})=D(\rho_{s,m,k}\|\sigma_{s,m,k}).
\end{equation}
These quantities are additive on arbitrary finite products of such states. The equalities and additivity also hold for arbitrary finite products of the flagged states $J_{s,m}$.
\end{appthm}
\begin{proof}
Write $W=\rho_{s,m,k}\sigma_{s,m,k}^{-1}$. This is well defined on the finite sector because every $\mu_h$ is positive. Since the two operators commute, $\log_2 W=\log_2\rho-\log_2\sigma$ on the resource support. The remaining task is to prove that this likelihood-ratio operator satisfies the partial-transpose norm constraint. Expanding $H_h$ about $z=1$ and using Vandermonde's identity gives
\begin{equation}\label{eq:K13}
 \frac{H_h(z)}{\binom Ah}
 =\sum_{j=0}^{h}\frac{h^{\underline j}(A-h)^{\underline j}}
 {j!A^{\underline j}}(z-1)^j.
\end{equation}
Thus, for a positive normalizing constant $c$,
\begin{equation}\label{eq:K14}
 W=c\sum_{j=0}^{q}\frac{(z-1)^j}{j!A^{\underline j}}C_j.
\end{equation}
Every coefficient is nonnegative. Equation~\eqref{eq:K9} gives $\operatorname{Tr}(\sigma C_j)=m^{\underline j}k^{\underline j}$. Since $\operatorname{Tr}(\sigma W)=1$, Lemma~\ref{lem:K1} implies
\begin{equation}\label{eq:K15}
 \|W^\Gamma\|_\infty\leq
 c\sum_{j=0}^{q}\frac{(z-1)^jm^{\underline j}k^{\underline j}}
 {j!A^{\underline j}}=1.
\end{equation}
The aligned condensate saturates this inequality. Lemma~E.1 then gives $R(\rho)\geq\operatorname{Tr}(\rho\log_2W)=D(\rho\|\sigma)$. The reverse inequalities follow from $R\leq E_R$ and separability of $\sigma$. Tensor products of the witnesses give the lower bound for products, while products of the comparisons give the upper bound.

Both measures obey the local-flags identity. For $R$, pinch a comparison into its flagged blocks $\omega_i$ and put $s_i=\|\omega_i^\Gamma\|_1$. Feasibility is $\sum_i s_i\leq1$. For flag probabilities $p_i$, relative entropy splits into the conditional costs and $\sum_i p_i\log_2(p_i/s_i)$. The latter has minimum zero at $s_i=p_i$. This proves the identity for finitely many flags. For countably many flags, the direct sum of optimal comparisons gives the upper bound. For the lower bound, retain finitely many flag outcomes and group the rest into one outcome. The finite-flags identity and nonnegativity bound the result below by the retained conditional costs. Increasing the retained set gives the full sum. Here every conditional cost is at most $\log_2d_{s,m}$. Applying the identity to Bob's flag proves the assertion for $J_{s,m}$. If $s=1$, Alice's sector has dimension one and both measures vanish.
\end{proof}

\subsection{Regularized Gaussian-state entanglement}

Set $f_s(m)=R(J_{s,m})=E_R(J_{s,m})$. We use the deterministic sector-density limit of Ref.~\cite{r16}, Appendix~C, Lemma~3:
\begin{equation}\label{eq:K16}
 \frac{f_s(m_s)}s\longrightarrow F_{\tau,v}(N)
 \quad\text{whenever}\quad\frac{m_s}s\longrightarrow N.
\end{equation}
Also $0\leq f_s(m)/s\leq g(m/s)$. The upper bound follows from the input Schmidt rank and $\log_2d_{s,m}\leq sg(m/s)$.

\begin{appcor}[Energy-constrained state measures]\label{cor:K3}
For every finite $N$ in the strictly noisy region,
\begin{equation}\label{eq:K17}
 R^\infty(\chi_N^{\mathcal G_{\tau,v}})
 =E_R^\infty(\chi_N^{\mathcal G_{\tau,v}})=F_{\tau,v}(N).
\end{equation}
\end{appcor}
\begin{proof}
Measure Alice's total photon count $M_s$ on $\chi_N^{\otimes s}$. It has a negative-binomial distribution with mean $sN$. Within a fixed count, the input squeezed-vacuum amplitudes are equal, so the conditional state is $J_{s,m}$. Local monotonicity, flags and a direct sum of the optimal separable comparisons imply
\begin{equation}\label{eq:K18}
 \begin{split}
 \mathbb E f_s(M_s)&\leq R(\chi_N^{\otimes s})\leq E_R(\chi_N^{\otimes s})\\
 &\leq\mathbb E f_s(M_s)+H(M_s).
 \end{split}
\end{equation}
Here $H$ is Shannon entropy. The last bound follows because the comparison commutes with number pinching and the entropy increase under that pinching is at most $H(M_s)$. Furthermore, $H(M_s)\leq g(sN)=O(\log s)$.

The law of large numbers gives $M_s/s\to N$ in probability. The deterministic assertion \eqref{eq:K16} is locally uniform: a violating sequence with densities tending to $N$ would contradict it. To control the tails, for $B>0$ use
\begin{equation}\label{eq:K19}
 \mathbb E\!\left[g(M_s/s)\mathbf1_{\{M_s/s>B\}}\right]
 \leq N\sup_{x>B}\frac{g(x)}x\longrightarrow0
\end{equation}
as $B\to\infty$. Thus $\mathbb E f_s(M_s)/s\to F_{\tau,v}(N)$. Divide Eq.~\eqref{eq:K18} by $s$. At $N=0$, and throughout the entanglement-breaking region, the channel state is separable.
\end{proof}

\section{Adaptive energy-constrained converse}\label{app:L}

We prove the converse under Eq.~\eqref{eq:energy_constraint} by bounding the one-use increase of regularized Rains entanglement for inputs correlated with both laboratories' memories. Iteration and concavity then give the adaptive bound.

\subsection{Rains properties and number pinching}

The Rains bound is convex, subadditive, lower semicontinuous and monotone under the free operations. A partially transposed free map contracts the trace norm of a Hermitian operator and therefore preserves the comparison set. Relative-entropy data processing proves monotonicity. Convexity of the comparison set and closure under tensor products prove convexity and subadditivity.

For local number pinching $\mathcal P$, with count outcome $J$ of finite entropy,
\begin{equation}\label{eq:L1}
 R(\omega)\leq R(\mathcal P\omega)+H(J).
\end{equation}
Choose a Rains comparison that is block diagonal under $\mathcal P$. Pinching any comparison preserves its feasibility and cannot increase its relative entropy from $\mathcal P\omega$. For a block-diagonal comparison $\sigma$, the relative-entropy identity is
\begin{equation}\label{eq:L2}
 D(\omega\|\sigma)=D(\omega\|\mathcal P\omega)+D(\mathcal P\omega\|\sigma).
\end{equation}
For the infinite-dimensional identity, choose increasing finite-rank projections that commute with the number projections and converge strongly to the identity. In each finite corner, expanding the logarithms proves the corresponding identity for the trace-corrected relative entropy $\operatorname{Tr}[X(\ln X-\ln Y)+Y-X]$. Its monotone finite-corner limit gives Eq.~\eqref{eq:L2}, including infinite values: the trace corrections cancel because pinching preserves the trace. This argument never subtracts two infinite entropies.

To bound the first term, purify $\omega$. Pinching the purification produces a mixture of orthogonal pure states with probabilities $p_j=\Pr(J=j)$. The relative entropy of the purification with respect to this mixture is $H(J)$. Tracing out the purifying system and using data processing therefore gives $D(\omega\|\mathcal P\omega)\leq H(J)$. Taking the infimum over block-diagonal $\sigma$ in Eq.~\eqref{eq:L2} proves Eq.~\eqref{eq:L1}. No finite-entropy assumption on the memories is needed. Only one collective count is recorded, so its cost grows logarithmically with the number of replicated modes at fixed mean energy.

For lower semicontinuity, use the weak-star topology on trace-class operators tested against compact operators. The Rains comparison set is weak-star closed within the trace-class unit ball. Positivity and the partial-transpose constraint pass to every finite local matrix corner; the uniform trace-norm bound on these corners gives a trace-class partial transpose with norm at most one. The comparison set is therefore weak-star compact, and sequentially compact on separable Hilbert spaces.

The required relative-entropy statement allows the comparison to lose trace in this limit. For a finite-rank projection $P$, define the normalised states $\rho[P]=P\rho P\oplus(1-\operatorname{Tr}P\rho)$ and $\sigma[P]=P\sigma P\oplus(1-\operatorname{Tr}P\sigma)$, with a one-dimensional tail flag. The comparison tail includes the missing normalisation mass, rather than only the trace of the discarded block. Then $D(\rho\|\sigma)=\sup_P D(\rho[P]\|\sigma[P])$. Indeed, these are trace-preserving compressions of $\rho\oplus0$ and $\sigma\oplus(1-\operatorname{Tr}\sigma)$. Data processing gives one inequality; their trace-norm convergence as $P\uparrow I$ and ordinary lower semicontinuity give the other. For fixed $P$, the padded states depend only on finite matrix entries, so finite-dimensional lower semicontinuity proves the claim for a trace-norm-convergent state and a weak-star-convergent comparison.

Now take a subsequence realising the lower limit of the Rains values and choose near-minimising comparisons. Compactness gives a weak-star-convergent subsubsequence with a feasible limit. Applying the preceding inequality and then minimising proves lower semicontinuity of $R$. No continuity of the regularised measure is assumed.

\begin{applemma}[One-use Rains increase]\label{lem:L1}
Let $\rho_{XAY}$ have finite $R(XA:Y)_\rho$ and transmitted mean $x=\operatorname{Tr}(\hat n_A\rho_A)<\infty$. Alice holds $XA$, Bob holds $Y$, and the channel output is $\omega_{XBY}=(\operatorname{id}_{XY}\otimes\mathcal G_{\tau,v})(\rho)$. Then
\begin{equation}\label{eq:L3}
 R^\infty(X:BY)_\omega\leq R^\infty(XA:Y)_\rho+F_{\tau,v}(x).
\end{equation}
\end{applemma}
\begin{proof}
First consider an $s$-mode transmitted block with total mean $sx$ and total input-count probabilities $p_m$. The channel restricted to $\mathcal H_{s,m}$ has an exact LOCC teleportation simulation using $J_{s,m}$~\cite{r15}, Appendix~B.1; see also Ref.~\cite{r16}, Appendix~B.2. This simulation also acts correctly on inputs correlated with both memories.

Total input-number pinching commutes through the phase-covariant channel as total output-number pinching. To see this, write pinching as averaging a common optical phase over all transmitted modes. Phase covariance carries that same average to the outputs, while the reference and memory systems are untouched. The input count can consequently be exposed to the sector simulator, while Eq.~\eqref{eq:L1} bounds the coherence removed from the output. Apply the sector simulator to each conditional input. Selective local monotonicity, flags, subadditivity and convexity give
\begin{equation}\label{eq:L4}
 R(\mathcal P_B\omega)\leq R(\rho)+\sum_m p_m f_s(m).
\end{equation}
Bob's total output count has mean $s(\tau x+v)$ and entropy at most $g(s(\tau x+v))$. Bounding the removed output coherence with Eq.~\eqref{eq:L1} yields
\begin{equation}\label{eq:L5}
 R(\omega)\leq R(\rho)+\sum_m p_m f_s(m)+g(s(\tau x+v)).
\end{equation}
Initially impose finite memory and input cutoffs. Remove the output cutoff first, replacing discarded output weight by vacuum. This local channel is phase covariant and cannot increase the mean or the Rains bound. Lower semicontinuity gives Eq.~\eqref{eq:L5} with infinite output support. Next remove input and memory cutoffs. Input vacuum replacement moves discarded probability only to $m=0$, where $f_s(0)=0$. Its resource average cannot increase and is bounded by $s\sum_m p_m g(m/s)\leq sg(x)$. Memory cutoffs leave the optical mean unchanged. Trace-norm convergence and lower semicontinuity complete these limits.

Apply Eq.~\eqref{eq:L5} to $s$ independent replicas of the fixed one-use input $\rho$. Their total input count $M_s$ is a sum of independent counts of finite mean $x$, even if $\rho_A$ has number coherence. Hence
\begin{equation}\label{eq:L6}
 \begin{split}
 \frac{R(\omega^{\otimes s})}{s}
 &\leq\frac{R(\rho^{\otimes s})}{s}+\mathbb E\frac{f_s(M_s)}s\\
 &\quad+\frac{g(s(\tau x+v))}{s}.
 \end{split}
\end{equation}
The middle term tends to $F_{\tau,v}(x)$ by Eq.~\eqref{eq:K16}, local uniformity and the tail estimate \eqref{eq:K19}. The last term vanishes. This proves Eq.~\eqref{eq:L3}. The replicas regularize a fixed one-use state. Since the bound holds for every such state, it applies at each adaptive use with its unconditional transmitted mean.
\end{proof}

\begin{appthm}[Mean-energy PPT converse]\label{thm:L2}
For the strictly noisy channels and every finite $N$,
\begin{equation}\label{eq:L7}
 Q_{\mathrm{PPT}}(\mathcal G_{\tau,v},N)\leq F_{\tau,v}(N).
\end{equation}
\end{appthm}
\begin{proof}
By Sec.~A\,6, it suffices to bound a distribution protocol that obeys Eq.~\eqref{eq:energy_constraint}. Iterate Lemma~\ref{lem:L1} through a fixed $r$-use protocol, whose output is $\zeta_r$. The initial Rains value is zero, and free operations do not increase $R^\infty$. Concavity and monotonicity of $F_{\tau,v}$~\cite{r16}, Appendix~D, Lemma~6, give
\begin{equation}\label{eq:L8}
 R^\infty(\zeta_r)\leq\sum_{i=1}^{r}F_{\tau,v}(x_i)\leq rF_{\tau,v}(N).
\end{equation}
Equation~\eqref{eq:L5} with $s=1$ keeps the unregularized Rains value finite after each use, so every invocation of the lemma is legitimate.

Let $K\geq2$ be the target dimension and $f=\operatorname{Tr}(\Phi_K\zeta_r)$. We may assume that $f\geq1/K$, since otherwise there is nothing to prove. Define
\begin{equation}\label{eq:L9}
 W=fK\Phi_K+\frac{(1-f)K}{K-1}(I-\Phi_K).
\end{equation}
Its partial transpose has eigenvalues $1$ and $(K+1-2Kf)/(K-1)$, both in $[-1,1]$. Lemma~E.1 on tensor powers therefore gives, with $h_2$ the binary entropy,
\begin{equation}\label{eq:L10}
 \begin{split}
 R^\infty(\zeta_r)&\geq\log_2K-h_2(f)-(1-f)\log_2(K-1)\\
 &\geq f\log_2K-h_2(f).
 \end{split}
\end{equation}
Boundary values follow by a limit. Combine this with Eq.~\eqref{eq:L8}, divide by $r$ and take $f\to1$. Entanglement-breaking channels have a separable measure-and-prepare simulation and zero PPT-assisted capacity.
\end{proof}

\section{A polynomial norm bound for correlation thresholds}\label{app:M}

Fix a photon sector with the decomposition in Eq.~\eqref{eq:K7}. For $0\leq h\leq q$, define the correlation-threshold projector $T_{s,m,k,h}$ and its separable acceptance probability $p_{s,m,k,h}$ by
\begin{equation}\label{eq:M1}
\begin{aligned}
T_{s,m,k,h}&=\sum_{a=h}^{q}\Pi_a,\\
p_{s,m,k,h}&=\operatorname{Tr}(\sigma_{s,m,k}T_{s,m,k,h})
=\sum_{a=h}^{q}D_a\mu_a.
\end{aligned}
\end{equation}
The separable acceptance probability $p_{s,m,k,h}$ is easy to evaluate from the sector spectrum and will have exponent $F_{\tau,v}(N)$ on typical sectors. The coding condition instead involves $\lVert T^\Gamma\rVert_\infty$. This appendix bridges the two. First, aligned condensates maximize the test expectation. Second, symmetry reduces the norm estimate to two modes per laboratory, whose dimensions grow only polynomially with photon number.

\begin{appthm}[Uniform polynomial threshold bound]\label{thm:M1}
For every $s\geq2$, $m,k\geq0$, and $0\leq h\leq\min(m,k)$,
\begin{equation}\label{eq:M2}
p_{s,m,k,h}\leq\lVert T_{s,m,k,h}^\Gamma\rVert_\infty
\leq(m+1)(k+1)p_{s,m,k,h}.
\end{equation}
For $m,k=O(s)$, this implies
\begin{equation}\label{eq:M3}
0\leq\frac1s\log_2\frac{\lVert T_{s,m,k,h}^\Gamma\rVert_\infty}{p_{s,m,k,h}}
\leq\frac{\log_2(m+1)+\log_2(k+1)}s.
\end{equation}
The rate penalty is $O(\log s/s)$, uniformly over thresholds and bounded photon-density ranges.
\end{appthm}

\subsection{A positive formula for condensate expectations}
The threshold is a polynomial in the pair operators $C_j$, but its coefficients alternate in sign. Bounding their absolute values term by term would discard the cancellations and can change the exponential rate. The following argument keeps the polynomial intact and expresses its derivative as an expectation of a nonnegative random variable.

\begin{applemma}[Condensate maximum]\label{lem:M2}
For unit vectors $e,f\in\mathbb C^s$,
\begin{equation}\label{eq:M4}
\langle\bar e;m|\langle f;k|T_{s,m,k,h}|\bar e;m\rangle|f;k\rangle
\leq p_{s,m,k,h}.
\end{equation}
Equality holds when $e=f$.
\end{applemma}
\begin{proof}
The case $h=0$ is immediate. For $h\geq1$, the threshold has the polynomial expansion
\begin{equation}\label{eq:M5}
\begin{aligned}
T_{s,m,k,h}&=\sum_{j=h}^{q}c_{h,j}C_j,\\
c_{h,j}&=\frac{(-1)^{j-h}(A-h-j)!}{(h-1)!(j-h)!\,j\,(A-h)!}.
\end{aligned}
\end{equation}
A verification is given below. Put $x=|\langle e,f\rangle|^2$. Annihilating $j$ pairs gives the condensate expectation
\begin{equation}\label{eq:M6}
f_h(x)=\sum_{j=h}^{q}c_{h,j}m^{\underline j}k^{\underline j}x^j,
~~ 0\leq x\leq1.
\end{equation}
Let $M=m-h$, $B=k-h$, and $D=A-2h=s+M+B-1$. Sample $M$ objects without replacement from $D$ objects, of which $B$ are marked. If $J$ is the number of marked objects sampled, then
\begin{equation}\label{eq:M7}
\Pr(J=r)=\frac{\binom Br\binom{D-B}{M-r}}{\binom DM},
~~ \mathbb E J^{\underline r}=\frac{M^{\underline r}B^{\underline r}}{D^{\underline r}}.
\end{equation}
The parameters are admissible because $D\geq M+B$. The factorial moment counts ordered marked tuples. Define the positive constant
\begin{equation}\label{eq:M8}
\beta_h=\frac{D!\,m^{\underline h}k^{\underline h}}{(h-1)!(A-h)!}.
\end{equation}
Differentiating and writing $j=h+r$ gives
\begin{equation}\label{eq:M9}
\begin{aligned}
f_h'(x)&=\beta_hx^{h-1}\sum_{r=0}^{\min(M,B)}
\frac{(-1)^r M^{\underline r}B^{\underline r}}{r!D^{\underline r}}x^r\\
&=\beta_hx^{h-1}\mathbb E(1-x)^J\geq0.
\end{aligned}
\end{equation}
Thus $f_h(x)\leq f_h(1)$. The threshold is invariant under conjugate passive mode rotations on Alice and ordinary rotations on Bob. Its aligned-condensate expectation is independent of the common direction. Averaging that direction as in Eq.~\eqref{eq:K9} gives $f_h(1)=p_{s,m,k,h}$.

To verify Eq.~\eqref{eq:M5}, evaluate it on $\Pi_a$. It vanishes for $a<h$. For $a\geq h$, put $M_0=a-h$, $B_0=A-a-h$, and $D_0=M_0+B_0=A-2h$. Equation~\eqref{eq:K5} gives its value as
\begin{equation}\label{eq:M10}
\frac{D_0!\,a^{\underline h}(A-a)^{\underline h}}{h!(A-h)!}
\sum_{r=0}^{M_0}\frac{h}{h+r}
\frac{(-1)^r M_0^{\underline r}B_0^{\underline r}}{r!D_0^{\underline r}}.
\end{equation}
Use the hypergeometric distribution \eqref{eq:M7} with population $D_0=M_0+B_0$. The sum is $h\int_0^1x^{h-1}\mathbb E(1-x)^J\,dx$. The beta integral and Vandermonde convolution give
\begin{equation}\label{eq:M11}
\sum_{r=0}^{M_0}\frac{h}{h+r}
\frac{(-1)^r M_0^{\underline r}B_0^{\underline r}}{r!D_0^{\underline r}}
=\frac{M_0!\,h!\,B_0!\,(D_0+h)!}{D_0!\,(M_0+h)!\,(B_0+h)!}.
\end{equation}
Explicitly, after the beta integral, apply the binomial identity
\begin{equation}\label{eq:M12}
\frac{\binom{M_0}{r}\binom{B_0}{r}}{\binom{h+r}{r}}
=\frac{M_0!h!}{(M_0+h)!}\binom{B_0}{r}\binom{M_0+h}{M_0-r}.
\end{equation}
Summing over $r$ proves Eq.~\eqref{eq:M11}. Its product with the prefactor in Eq.~\eqref{eq:M10} is one. This proves the polynomial identity on every ladder space.
\end{proof}

\subsection{Reduction of the norm to two modes}
The full sector dimensions grow rapidly with the number of modes, so a direct dimension-times-expectation estimate is too costly. Passive-mode invariance gives a smaller space that still contains a vector for every eigenvalue of the partially transposed operator. The resulting estimate depends on $m$ and $k$, but not exponentially on $s$.

\begin{applemma}[Two-mode compression]\label{lem:M3}
Let $X\geq0$ act on $\operatorname{Sym}^m\mathbb C^s\otimes\operatorname{Sym}^k\mathbb C^s$, with $s\geq2$, and be invariant under $\bar U^{\otimes m}\otimes U^{\otimes k}$ for every unitary $U$. If every condensate-pair expectation of $X$ is at most $p$, then
\begin{equation}\label{eq:M13}
\lVert X^\Gamma\rVert_\infty\leq(m+1)(k+1)p.
\end{equation}
\end{applemma}
\begin{proof}
After partial transpose, $X^\Gamma$ commutes with the same mode rotation on both factors; complex conjugating that rotation does not change this statement. The symmetric-power decomposition is
\begin{equation}\label{eq:M14}
\operatorname{Sym}^m\mathbb C^s\otimes\operatorname{Sym}^k\mathbb C^s
=\bigoplus_{\ell=0}^{q}V_{(m+k-\ell,\ell)}.
\end{equation}
This is the multiplicity-free two-factor Pieri decomposition. It follows directly at the level of characters. Write $h_r$ for the character of $\operatorname{Sym}^r\mathbb C^s$, with $h_0=1$ and $h_{-1}=0$, and $s_{(a,b)}$ for the irreducible character of highest weight $(a,b)$. The two-row Jacobi--Trudi identity gives $s_{(a,b)}=h_ah_b-h_{a+1}h_{b-1}$. Summing over $(a,b)=(m+k-\ell,\ell)$ telescopes to $h_mh_k$, since $q=\min(m,k)$. Thus every listed irreducible occurs exactly once and there are no further summands. In the polynomial representation, its highest-weight vector is
\begin{equation}\label{eq:M15}
v_\ell=(a_1^\dagger b_2^\dagger-a_2^\dagger b_1^\dagger)^\ell
(a_1^\dagger)^{m-\ell}(b_1^\dagger)^{k-\ell}|0\rangle.
\end{equation}
Its weight is $(m+k-\ell,\ell)$; all raising operators $a_i^\dagger a_j+b_i^\dagger b_j$ with $i<j$ annihilate it. Its orbit generates the corresponding summand. Every summand therefore has a representative using only the first two modes in each laboratory. This explicit representative is the only representation-theoretic property needed for the norm reduction: all possible spectral values survive the same two-mode compression.

Schur's lemma makes $X^\Gamma$ scalar on every summand. Let $E$ embed the two-mode symmetric spaces in the full spaces and set $B=E^\dagger XE\geq0$. The embedding is real in the number basis, so $B^\Gamma=E^\dagger X^\Gamma E$. Every eigenvalue of $X^\Gamma$ occurs on a vector \eqref{eq:M15} and hence also in this compression. Conversely, compression cannot increase operator norm. For the Hilbert--Schmidt norm $\lVert Z\rVert_2=(\operatorname{Tr}Z^\dagger Z)^{1/2}$, this gives
\begin{equation}\label{eq:M16}
\lVert X^\Gamma\rVert_\infty=\lVert B^\Gamma\rVert_\infty
\leq\lVert B^\Gamma\rVert_2=\lVert B\rVert_2\leq\operatorname{Tr}B.
\end{equation}
The Hilbert--Schmidt equality uses invariance under partial transpose; the final inequality uses positivity.

Two-mode condensates resolve the identity:
\begin{equation}\label{eq:M17}
(m+1)\int_{\lVert e\rVert=1}|e;m\rangle\langle e;m|\,de
=I_{\operatorname{Sym}^m\mathbb C^2}.
\end{equation}
This follows from irreducibility and the trace; the measure is normalized on unit vectors. Taking the product of the two resolutions gives
\begin{equation}\label{eq:M18}
\begin{aligned}
\operatorname{Tr}B&=(m+1)(k+1)\iint
\langle\bar e;m|\langle f;k|B|\bar e;m\rangle|f;k\rangle\,de\,df\\
&\leq(m+1)(k+1)p.
\end{aligned}
\end{equation}
Together with Eq.~\eqref{eq:M16}, this proves the claim.
\end{proof}

\begin{proof}[Proof of Theorem~\ref{thm:M1}]
Apply Lemma~\ref{lem:M3} to $X=T_{s,m,k,h}$ and use Lemma~\ref{lem:M2}. For the lower bound, choose a real unit vector $e$. Partial transpose fixes the aligned product projector. Its expectation in $T^\Gamma$ is $p_{s,m,k,h}$, so the operator norm is at least that value. Taking logarithms proves Eq.~\eqref{eq:M3}.
\end{proof}

\section{Energy-constrained achievability and channel boundaries}\label{app:N}
We first prove achievability for $N>0$ in the strictly noisy non-entanglement-breaking region \eqref{eq:K1}. All quantities $y,\ell_0,t,u,w$ use the Gaussian-channel parameters defined in Sec.~III; thermal loss is obtained by setting $\tau=\eta$ and $v=(1-\eta)n$.

The construction has three steps. Establish where the actual resource concentrates in photon counts and correlation label; evaluate how rarely a separable comparison passes a nearby threshold; then use the polynomial norm estimate to turn this threshold into a PPT channel. No decoding operation that measures the hidden normal modes or announces the correlation outcome is assumed.

\subsection{Concentration of the correlation label}
Alice prepares $s$ independent two-mode squeezed vacua with mean $N$ and transmits one mode of each. The output is $\chi_N^{\otimes s}$, and every input has mean exactly $N$. Let $M_s,K_s$ be the total reference and output photon counts. Let $H_s$ be the ladder label in their joint sector. These three observables commute and obey the following limits in probability:
\begin{equation}\label{eq:N1}
\frac{M_s}s\longrightarrow N,~~
\frac{K_s}s\longrightarrow y,~~
\frac{H_s}s\longrightarrow t.
\end{equation}
The first two assertions are laws of large numbers for thermal counts, whose single-copy variances are $N(N+1)$ and $y(y+1)$.

For the third, put $c_i=a_ib_i$ and $c=\langle c_i\rangle=\sqrt{\tau N(N+1)}$. With $\Delta=\sum_i(c_i-c)$,
\begin{equation}\label{eq:N2}
C_+C_-=s^2c^2+sc(\Delta+\Delta^\dagger)+\Delta^\dagger\Delta.
\end{equation}
Independence and centering give $\langle(\Delta+\Delta^\dagger)^2\rangle=O(s)$. In $\langle(\Delta^\dagger\Delta)^2\rangle$, any index appearing only once has zero expectation. Only $O(s^2)$ index patterns remain, and their moments are bounded by fixed single-copy Gaussian moments. Thus
\begin{equation}\label{eq:N3}
\frac{\langle C_+C_-\rangle}{s^2}=c^2+O(s^{-1}),~~
\operatorname{Var}(C_+C_-)=O(s^3).
\end{equation}
The variance bound also follows directly by applying $(A+B)^2\leq2A^2+2B^2$ to the two centered self-adjoint terms. Chebyshev's inequality implies $C_+C_-/s^2\to c^2$. On a joint number sector,
\begin{equation}\label{eq:N4}
C_+C_-=H_s(s+M_s+K_s-1-H_s).
\end{equation}
For $0\leq H_s\leq\min(M_s,K_s)$, the right side is strictly increasing. Its smaller-root inverse, scaled by $s$, is continuous. The first two convergences therefore imply the third. Fixed positive deviations have probability $O(s^{-1})$, uniformly when $N$ ranges over a compact neighborhood of a positive value.

\subsection{The separable tail exponent}
The resource typically has correlation density $t$. The separable comparison assigns exponentially small probability to a threshold just below that value. Its decay rate determines the distillation rate, up to the polynomial factor in Theorem~\ref{thm:M1}. For real densities $M,Y\geq0$ and $0\leq a\leq\min(M,Y)$, define
\begin{equation}\label{eq:N5}
\begin{aligned}
E_{M,Y}(a)={}&\phi(a)+\phi(1+M+Y-a)\\
&-\phi(M)-\phi(Y)-g(M-a)-g(Y-a).
\end{aligned}
\end{equation}
Uniform Stirling estimates in Eqs.~\eqref{eq:K4} and \eqref{eq:K8} give
\begin{equation}\label{eq:N6}
-\frac1s\log_2(D_j\mu_j)=E_{m/s,k/s}(j/s)+O(\log s/s)
\end{equation}
on any bounded range of $m/s,k/s$, including the label endpoints. To see that dimension subtraction causes no exponential loss, write
\begin{equation}\label{eq:N7}
\frac{D_j}{d_{s,m-j}d_{s,k-j}}
=1-\frac{(m-j)(k-j)}{(s+m-j-1)(s+k-j-1)}.
\end{equation}
The right-hand side is bounded away from zero on those ranges for large $s$.

Put $L_{M,Y}=1+M+Y$. In the interior,
\begin{equation}\label{eq:N8}
E_{M,Y}'(a)=\log_2\frac{a(1+M-a)(1+Y-a)}{(L_{M,Y}-a)(M-a)(Y-a)}.
\end{equation}
The numerator minus the denominator inside this ratio equals $(L_{M,Y}-2a)\{a(L_{M,Y}-a)-MY\}$. Thus the exponent increases when $a(L_{M,Y}-a)>MY$. At the physical point,
\begin{equation}\label{eq:N9}
t(\ell_0-t)-Ny=N(\tau-v)>0.
\end{equation}
Choose a small $\delta>0$ and set $\alpha=t-\delta>0$, preserving the strict inequality. Then choose $\epsilon>0$ so that it remains strict for $|M-N|\leq\epsilon$, $|Y-y|\leq\epsilon$, and $\alpha<\min(M,Y)$ throughout that window.

Let $h_s=\lfloor s\alpha\rfloor$. The exponent increases over the accepted tail. There are only $O(s)$ labels, so Eq.~\eqref{eq:N6} gives, uniformly in the window,
\begin{equation}\label{eq:N10}
-\frac1s\log_2p_{s,m,k,h_s}=E_{m/s,k/s}(\alpha)+O(\log s/s).
\end{equation}
The first tail term supplies a lower bound on the probability; the number of terms times the largest exponential order supplies the upper bound. Rounding $h_s$ contributes only the stated error.

\subsection{An admissible distillation map}
Fix a rate $0<R<F_{\tau,v}(N)$. Since $E_{N,y}(t)=F_{\tau,v}(N)$, continuity allows the preceding $\delta,\epsilon$ to be chosen so that
\begin{equation}\label{eq:N11}
\inf_{\substack{|M-N|\leq\epsilon\\|Y-y|\leq\epsilon}}E_{M,Y}(\alpha)>R.
\end{equation}
Let $Q_s$ be the direct sum of $T_{s,m,k,h_s}$ over this number window, and zero in all other sectors. It is a finite-rank effect with $0\leq Q_s\leq I$. Partial transpose preserves the local number blocks. By Theorem~\ref{thm:M1} and Eq.~\eqref{eq:N10},
\begin{equation}\label{eq:N12}
\lVert Q_s^\Gamma\rVert_\infty\leq2^{-sR}
\quad\text{for all sufficiently large }s.
\end{equation}
The polynomial factor is absorbed by the strict rate margin. Concretely, if the tail exponent exceeds $R$ by $\Delta>0$ throughout the count window, the norm is at most $2^{-s(R+\Delta)+O(\log s)}$. For large enough $s$, the linear margin $s\Delta$ dominates the logarithmic correction. Equation~\eqref{eq:N1} also gives $\operatorname{Tr}(Q_s\chi_N^{\otimes s})\to1$.

Set $d_s=\lfloor2^{sR}\rfloor$, which is at least two for large $s$. Use the standard-form channel
\begin{equation}\label{eq:N13}
\Lambda_s(X)=\operatorname{Tr}(Q_sX)\Phi_{d_s}
+\operatorname{Tr}((I-Q_s)X)\frac{I-\Phi_{d_s}}{d_s^2-1}.
\end{equation}
It is normal, completely positive and trace preserving. By Lemma~D.1, complete PPT preservation is exactly $-I/d_s\leq Q_s^\Gamma\leq I/d_s$. The two output swap eigenspaces yield the positive coefficient operators
\begin{equation}\label{eq:N14}
\frac{I+d_sQ_s^\Gamma}{d_s(d_s+1)},~~
\frac{I-d_sQ_s^\Gamma}{d_s(d_s-1)}.
\end{equation}
Equation~\eqref{eq:N12} ensures these inequalities. The output fidelity is $\operatorname{Tr}(Q_s\chi_N^{\otimes s})\to1$, and no test-outcome flag is supplied.

Thus $D_{\mathrm{PPT}}(\chi_N)\geq R$. Teleportation (Appendix~\ref{app:A}) turns the distilled ebits into a code at the same rate, using no further channel uses or photons, so $Q_{\mathrm{PPT}}(\mathcal G_{\tau,v},N)\geq R$ as well. Let $R$ increase to $F_{\tau,v}(N)$. The matching channel converse is Theorem~L.2. For states, Eq.~\eqref{eq:L10} and Rains monotonicity give $D_{\mathrm{PPT}}(\chi_N)\leq R^\infty(\chi_N)$, and Corollary~K.3 completes the equalities. This proves Theorems~4 and 5 in the strictly noisy region.

The limit order is explicit: fix the channel and $N$, choose a rate margin, increase $s$ and finally remove the margin. Only after this theorem is proved may $N$ be increased to recover the unrestricted PLOB expressions.

\subsection{A hard total-photon budget}
The same rate can be attained by a deterministic fixed-number input. Put $m_s=\lfloor sN\rfloor$ and send half of a maximally entangled state on $\mathcal H_{s,m_s}$. To prove concentration, condition $\chi_{m_s/s}^{\otimes s}$ on $M_s=m_s$. Its conditional state is $J_{s,m_s}$. The negative-binomial mass at its mean satisfies
\begin{equation}\label{eq:N15}
\Pr(M_s=m_s)=d_{s,m_s}\frac{(m_s/s)^{m_s}}{(1+m_s/s)^{s+m_s}}
=\Theta(s^{-1/2})
\end{equation}
for $N>0$, by Stirling's formula. Let $B_s$ be the event that the output count or correlation label lies outside the acceptance window. The concentration bounds are uniform for input means near $N$, so
\begin{equation}\label{eq:N16}
\Pr(B_s\mid M_s=m_s)\leq\frac{\Pr(B_s)}{\Pr(M_s=m_s)}
=O(s^{-1/2})\longrightarrow0.
\end{equation}
The count and correlation centers converge to those of the preceding code because $m_s/s\to N$. The same threshold and norm bound therefore achieve every rate below $F_{\tau,v}(N)$. Conditioning is only a proof device: the protocol prepares the fixed-number input directly. Its entire transmitted block has at most $sN$ photons.

\subsection{Zero energy, quantum-limited and identity cases}
At $N=0$, every transmitted mode is vacuum and is uncorrelated with the memories. The output can be prepared locally by Bob, so the capacity is zero. If $v\geq\tau$, the channel is entanglement breaking and all the quantities in Theorem~5 vanish.

Consider next $v=\max\{0,\tau-1\}$. Every nonzero conditional resource is
\begin{equation}\label{eq:N17}
\rho_{s,m,k}=\frac{\Pi_q}{D_q},~~ q=\min(m,k).
\end{equation}
Pure loss has $k\leq m$, while a quantum-limited amplifier has $k\geq m$. These states also follow by taking the $z\to\infty$ limit of Eq.~\eqref{eq:K7}; the highest ladder label dominates, and continuity of the physical channel identifies the limit.

The comparison eigenvalue $\mu_q$ remains as in Eq.~\eqref{eq:K8}. The operator $W=\Pi_q/(D_q\mu_q)$ is a positive multiple of $C_q$, so Lemma~K.1 and its aligned-condensate saturation give $\lVert W^\Gamma\rVert_\infty=1$. If a strictly positive witness is needed, use $(W+\varepsilon I)/(1+\varepsilon)$ and let $\varepsilon\downarrow0$. The exact sector Rains cost is therefore $-\log_2(D_q\mu_q)$, again additive.

For pure loss, this cost and its density are
\begin{equation}\label{eq:N18}
\begin{aligned}
-\log_2(D_q\mu_q)&=\log_2d_{s,m}-\log_2d_{s,m-k},\\
\frac{f_s(m_s)}s&\longrightarrow g(N)-g((1-\tau)N).
\end{aligned}
\end{equation}
The output count concentrates at $k/s\to\tau N$. For a quantum-limited amplifier,
\begin{equation}\label{eq:N19}
\begin{aligned}
-\log_2(D_q\mu_q)&=\log_2d_{s,k}-\log_2d_{s,k-m},\\
\frac{f_s(m_s)}s&\longrightarrow g(\tau N+\tau-1)-g((\tau-1)(N+1)).
\end{aligned}
\end{equation}
Here $k/s\to\tau N+\tau-1$. The loss count is binomial; the amplifier's added count is negative binomial. Both have variance $O(s)$ when $m_s=O(s)$. Stirling estimates, count concentration and the bound $f_s(m)\leq sg(m/s)$ justify these density limits.

The sector simulation and the Rains-increase argument of Appendix~\ref{app:L} apply unchanged, using these boundary costs. The energy functions are concave and nondecreasing~\cite{r15}. The boundary-sector witnesses and density limits also give the state sandwich \eqref{eq:K18}. Its proof uses only local number pinching, the exact sector costs and the bound $f_s(m)\leq sg(m/s)$, all of which remain valid here. Dividing by $s$ gives the same upper bound on $E_R^\infty(\chi_N)$ at these boundaries. Gaussian hashing achieves the two displayed limits by LOCC~\cite{r8,PGPBL}. To apply the finite-dimensional hashing inequality, locally project each output copy onto finite photon-number subspaces and announce joint success. The thermal marginal tails make the success probability tend to one as the cut-off increases. The normalised conditional states converge in trace norm and have uniformly bounded mean photon numbers. Entropy continuity under this finite-mode energy bound therefore gives convergence of the joint and marginal entropies, and hence of both coherent informations. For each fixed cut-off, hashing the successful copies gives the conditional rate multiplied by the success probability. Taking the coding limit first and then increasing the cut-off recovers the stated rates without changing the transmitted energy. This proves the boundary capacity and state equalities. At $\tau=1,v=0$, the channel is the identity. The same truncation of the pure two-mode squeezed state, followed by finite-dimensional pure-state concentration~\cite{r47}, gives the common rate $g(N)$ because its marginal entropy converges to $g(N)$. Thus all physical parameters in Theorems~4 and 5 are covered.

\section{Capacity benchmarks and the sharp norm question}\label{app:O}
This appendix collects the quantum-limited formulas and compares the exact capacities with simpler achievable constructions. It also isolates a stronger finite-block norm statement that remains open. A sharper norm could improve finite-block performance without changing the asymptotic capacity.

\subsection{Quantum-limited capacities and LOCC benchmarks}
At $n=0$, Theorem~1 recovers the pure-loss capacity, which is already attained by LOCC~\cite{r11}:
\begin{equation}\label{eq:O1}
Q_{\mathrm{PPT}}(\mathcal L_{\eta,0})=Q_2(\mathcal L_{\eta,0})=-\log_2(1-\eta).
\end{equation}
At $n=0$, the environment is in the vacuum and the channel is a quantum-limited amplifier. Theorem~2 recovers its known LOCC capacity~\cite{r11}:
\begin{equation}\label{eq:O2}
Q_{\mathrm{PPT}}(\mathcal A_{\kappa,0})=Q_2(\mathcal A_{\kappa,0})=-\log_2(1-\kappa^{-1}).
\end{equation}
For thermal loss, define the reverse-coherent-information benchmark $I_{\mathrm{RC}}(\eta,n)$ and write $[z]_+=\max\{0,z\}$ for a real number $z$. The benchmark and the resulting bounds are
\begin{align}
I_{\mathrm{RC}}(\eta,n)&=-\log_2(1-\eta)-g(n),\label{eq:O3}\\
[I_{\mathrm{RC}}(\eta,n)]_+&\leq Q_2(\mathcal L_{\eta,n})\leq Q_{\mathrm{PPT}}(\mathcal L_{\eta,n}).\label{eq:O4}
\end{align}
When $I_{\mathrm{RC}}$ is nonnegative, the improvement over this benchmark is $n\log_2(1/\eta)$. Distillability of non-PPT Gaussian states allows positive LOCC rates beyond regions certified by that expression~\cite{r48}.

Let $Q_2(\mathcal N,N)$ denote the same energy-constrained capacity with LOCC assistance. At zero thermal noise, Theorem~4 gives the known LOCC value~\cite{r15}:
\begin{equation}\label{eq:O5}
\begin{aligned}
Q_{\mathrm{PPT}}(\mathcal L_{\eta,0},N)&=Q_2(\mathcal L_{\eta,0},N)\\
&=g(N)-g((1-\eta)N).
\end{aligned}
\end{equation}
At the quantum-limited amplifier boundary, the common PPT and LOCC rate~\cite{r15} is
\begin{equation}\label{eq:O6}
\begin{aligned}
Q_{\mathrm{PPT}}(\mathcal A_{\kappa,0},N)&=Q_2(\mathcal A_{\kappa,0},N)\\
&=g(\kappa N+\kappa-1)-g((\kappa-1)(N+1)).
\end{aligned}
\end{equation}
At the noiseless additive endpoint, the identity channel has $Q_{\mathrm{PPT}}(\mathcal Z_0,N)=g(N)$. This is finite at fixed $N$, whereas its unrestricted capacity is infinite. For all three noisy families, taking $N\to\infty$ after the coding limit recovers the capacities in Sec.~II.

\subsection{Error-photon thresholds and hashing}
For thermal loss, keep $N,y,u,w$ as defined in Sec.~III. For any squeezer gain $H>1$, define $\theta_H=(H-1)/H$. The second marginal of $S_H^\dagger\chi_NS_H$ is thermal with mean
\begin{equation}\label{eq:O7}
\begin{aligned}
e_H={}&Hy+(H-1)(N+1)\\
&-2\sqrt{H(H-1)\eta N(N+1)}.
\end{aligned}
\end{equation}
This marginal is independent across copies and determines the threshold fidelity. If $e_H<H-1$, the threshold construction yields
\begin{equation}\label{eq:O8}
B_H(e_H)=\log_2H-g(e_H)+e_H\log_2\frac{H}{H-1}.
\end{equation}
This follows from the exact partial-transpose identity of Appendix~\ref{app:C} and the norm and fidelity bounds of Appendix~\ref{app:D}.

With $[x]_+=\max\{x,0\}$, the best of the two Gaussian hashing directions gives
\begin{equation}\label{eq:O9}
R_{\mathrm{hash}}(N)=\bigl[\max\{g(N),g(y)\}-g(u)-g(w)\bigr]_+.
\end{equation}
Consequently,
\begin{equation}\label{eq:O10}
\max\left\{R_{\mathrm{hash}}(N),\sup_{\substack{H>1\\e_H<H-1}}B_H(e_H)\right\}
\leq F_{\eta,n}(N).
\end{equation}
Time sharing can improve these benchmarks to their concave envelope; the sector protocol already attains the exact concave capacity function.

For example, at $\eta=0.8$, $n=0.5$, and $N=1$, the rates in ebits per use are
\begin{equation}\label{eq:O11}
\begin{array}{lr}
\text{Gaussian hashing}&\simeq0.4090\\
\text{Error threshold at }H\simeq8.175&\simeq0.5331\\
\text{Exact sector capacity }F_{\eta,n}(N)&\simeq0.6224.
\end{array}
\end{equation}
The middle row evaluates the displayed gain. Diagonalizing first gives $H\simeq1.523$ and the smaller threshold rate $\simeq0.1836$.

\subsection{A sharper finite-block equality}
The stronger identity
\begin{equation}\label{eq:O12}
\lVert T_{s,m,k,h}^\Gamma\rVert_\infty=p_{s,m,k,h}
\end{equation}
is conjectured. Theorem~\ref{thm:M1} proves the lower bound and a polynomial upper bound.

For an explicit form of the remaining question, let $\ell$ label the irreducible summands in Eq.~\eqref{eq:M14}. The eigenvalue of $C_j^\Gamma$ on the $\ell$th summand is
\begin{equation}\label{eq:O13}
b_{\ell,j}=(j!)^2\sum_{r=0}^{\min(\ell,j)}(-1)^r
\binom\ell r\binom{m-\ell}{j-r}\binom{k-\ell}{j-r},
\end{equation}
where out-of-range binomial coefficients vanish. To derive it, second-quantize the two-species matrix with diagonal entries one and off-diagonal entries $\sqrt x$, then compress to species counts $(m,k)$. Expansion gives $\sum_jC_j^\Gamma x^j/(j!)^2$. The two-species representation paired with the mode summand $\ell$ is $\det^{\ell}\otimes\operatorname{Sym}^{m+k-2\ell}\mathbb C^2$. After removing the determinant weight, compression selects the weight $(m-\ell,k-\ell)$. Its diagonal matrix element therefore gives
\begin{equation}\label{eq:O14}
(1-x)^\ell\sum_{j=0}^{q-\ell}\binom{m-\ell}{j}\binom{k-\ell}{j}x^j.
\end{equation}
Matching coefficients proves Eq.~\eqref{eq:O13}.

For $h\geq1$, the conjecture is therefore the finite family of inequalities
\begin{equation}\label{eq:O15}
\left|\sum_{j=h}^{q}c_{h,j}b_{\ell,j}\right|\leq\sum_{a=h}^{q}D_a\mu_a,
~~ 0\leq\ell\leq q.
\end{equation}
The $\ell=0$ component gives equality without the absolute value. The cases $h=0$ and $h=q$ follow from the identity operator and the partial-swap bound. Intermediate thresholds, finite-error refinements and optimal noisy-channel LOCC rates remain open.

\end{document}